\documentclass[sigconf]{acmart}

\AtBeginDocument{%
  \providecommand\BibTeX{{%
    \normalfont B\kern-0.5em{\scshape i\kern-0.25em b}\kern-0.8em\TeX}}}

\setcopyright{acmcopyright}
\copyrightyear{2019}
\acmYear{2019}
\acmDOI{}

\acmConference[PPDP '19]{The 21th International Symposium on Principles and Practice of Declarative Programming (PPDP '19)}{October 7--9, 2019}{Porto, Portugal}
\acmBooktitle{The 21th International Symposium on Principles and Practice of Declarative Programming (PPDP '19), October 7--9, 2019, Porto, Portugal}
\acmPrice{00}
\acmISBN{}

\usepackage{listings}
\usepackage{amssymb}
\usepackage{algorithm}
\usepackage{mathpartir}
\usepackage{soul}

\usepackage{mathtools}
\usepackage{mathrsfs}
\usepackage{algpseudocode}
\usepackage{syntax}
\usepackage{semantic}
\usepackage{tikz}
\usepackage[inline]{enumitem}

\newcommand{\eqdef}{\mathrel{\overset{\makebox[0pt]{\mbox{\normalfont\tiny\sffamily def}}}{=}}}

\newcommand{\dcup}{\;\dot\cup\;}

\newcommand{\locset}{\mathit{Loc}}
\newcommand{\storelat}{\mathit{Store}}
\newcommand{\prop}{\mathit{Prop}}
\newcommand{\csp}{\mathit{CSP}}
\newcommand{\solve}{\mathit{solve}}
\newcommand{\ST}{\mathit{ST}}
\newcommand{\spt}{\mathit{st}}
\newcommand{\antiL}{\mathscr{A}(L)}
\newcommand{\nameset}{\mathit{Name}}
\newcommand{\rw}{\mathit{rw}}

\newcommand{\xrightarrowdbl}[2][]{%
  \xrightarrow[#1]{#2}\mathrel{\mkern-14mu}\rightarrow
}

\lstdefinelanguage{bonsai}{
  keywords = {when,let,pause,loop,par,space,||,<-,let,in,mut,fn,end,exit,trap,for,in,otherwise,by,single_space,single_time,world_line,transient,bot,top,pre,await,each,clock,then,if,run,up,universe,prune,flow,else,nothing,public,protected,private,class,interface,extends,implements,proc,new,static,ref,module,read,write,readwrite,weak,return,abort,suspend,not,or,and,with,stop,true,false,unknown,int,void,boolean}
}

\reservestyle{\bonsai}{\textbf}
\bonsai{when[when\;],nothing[nothing\;],loop[loop\;],exit[exit\;],pause[pause\;],par[par\;],space[space\;],end[\;end],hpar[\;\texttt{||}\;],trap[trap\;],frontier[frontier\;],let[let\;],in[\;in\;],boot[\;boot],then[\;then\;],dia[\;\texttt{<>}\;],upar[\;\texttt{|}$\Delta$\texttt{|}\;],ipar[\;\texttt{<}$\Delta$\texttt{>}\;],square[\;$\square$\;],seq[\;;\;],prune[prune\;],pauseup[pause up\;],for[for],abort[abort\;],tell[\;\texttt{<-}\;],suspend[suspend\;],universe[universe\;],singletime[single\texttt{\_}time\;],singlespace[single\texttt{\_}space\;],worldline[world\texttt{\_}line\;],stop[stop\;],presy[pre],pre[pre\;],bot[bot],top[top],read[read\;],write[write\;],readwrite[readwrite\;],readx[read],writex[write],readwritex[readwrite],weak[weak\;],proc[proc\;],flow[flow\;],else[\;else\;],new[new\;],ask[\;\texttt{|=}\;],signal[signal\;],present[present\;],singletimex[single\texttt{\_}time],singlespacex[single\texttt{\_}space],worldlinex[world\texttt{\_}line],pauseupx[pause up],pausex[pause],stopx[stop],not[not\;],tor[\;or\;],tand[\;and\;],unknown[unknown],true[true],false[false],tellcc[tell],sumcc[$\sum_{i \in n}^n$],thencc[\;then\;],ask[ask],with[with]}

\definecolor{ForestGreen}{rgb}{0.0, 0.65, 0.31}

\begin{document}
\mathligsoff

\title{Spacetime Programming}         
\subtitle{A Synchronous Language for Composable Search Strategies}  


\author{Pierre Talbot}
\orcid{0000-0001-9202-4541}             
\affiliation{
  \position{}
  \department{}             
  \institution{Universit\'e de Nantes}           
  \streetaddress{2, rue de la Houssini\`ere}
  \city{Nantes}
  \state{}
  \postcode{44000}
  \country{France}                   
}
\email{pierre.talbot@univ-nantes.fr}         

\begin{abstract}
Search strategies are crucial to efficiently solve constraint satisfaction problems.
However, programming search strategies in the existing constraint solvers is a daunting task and constraint-based languages usually have compositionality issues.
We propose spacetime programming, a paradigm extending the synchronous language \textsf{Esterel} and \textit{timed concurrent constraint programming} with backtracking, for creating and composing search strategies.
In this formalism, the search strategies are composed in the same way as we compose concurrent processes.
Our contributions include the design and behavioral semantics of spacetime programming, and the proofs that spacetime programs are deterministic, reactive and extensive functions.
Moreover, spacetime programming provides a bridge between the theoretical foundations of constraint-based concurrency and the practical aspects of constraint solving.
We developed a prototype of the compiler that produces search strategies with a small overhead compared to the hard-coded ones.
\end{abstract}

\ccsdesc[500]{Theory of computation~Process calculi}
\ccsdesc[500]{Theory of computation~Timed and hybrid models}
\ccsdesc[300]{Theory of computation~Program semantics}
\ccsdesc[500]{Computing methodologies~Discrete space search}
\ccsdesc[300]{Computing methodologies~Concurrent programming languages}

\keywords{synchronous programming, concurrent constraint programming, constraint satisfaction problem, search strategy}

\maketitle


\section{Introduction}

Constraint programming is a powerful paradigm to model problems in terms of constraints over variables.
This declarative paradigm solves many practical problems including scheduling, vehicle routing or biology problems~\citep{handbook-cp}, as well as more unusual problems such as in musical composition~\citep{truchet-constraint-2011}.
Constraint programming describes \textit{what} the problem is, whereas procedural approaches describe \textit{how} a problem is solved.
The programmer declares the constraints of its problem, and relies on a generic constraint solver to obtain a solution.

A constraint satisfaction problem (CSP) is a couple $\langle d, C \rangle$ where $d$ is a function mapping variables to sets of values (the domain) and $C$ is a set of constraints on these variables.
The goal is to find a solution: a set of singleton domains such that every constraint is satisfied.
For example, given the CSP $\langle \{x \mapsto \{1,2,3\}, y \mapsto \{1,2,3\}\}, \{x > y, x \neq 2\} \rangle$, a solution is $\{x \mapsto 3, y \mapsto 1\}$.

The solving procedure usually interleaves two steps: propagation and search.
Propagation removes values from the domains that do not satisfy at least one constraint.
The search step makes a choice when propagation cannot infer more information and backtracks to another choice if the former one did not lead to a solution.
The successive interleaving of choices and backtracks lead to the construction of a search tree that can be explored with various search strategies.
In this paper, the term ``search strategy'' takes the broad sense of any procedure that describes how a CSP is solved.

In order to attain reasonable efficiency, the programmer must often customize the search strategy \textit{per problem}~\citep{beck-solution-guided-2007,Simonis:2008:SSR:1431540.1431546,teppan-quickpup:-2012}.
However, to program a search strategy in a constraint solver is a daunting task that requires expertise and good understanding on the solver's intrinsics.
This is why various language abstractions emerged to ease the development of search strategies~\citep{van-hentenryck-constraint-chip-1989,salsa-2002,van-hentenryck-contraint-based-2005,martinez-search-2015}.

One of the remaining problems of search languages is the compositionality of search strategies: how can we easily combine two strategies and form a third one?
Compositionality is important to build a collection of search strategies reusable across problems.
To cope with this compositionality issue, we witness a growing number of proposals based on functional programming~\citep{mcp}, constraint logic programming~\citep{Schrijvers:2014:TMS:2608851.2608962}, and search combinators~\citep{search2013}.
However, a recurring issue in these approaches is the difficulty to share information among strategies; we discuss this drawback and others in Section~\ref{related-work}.

We propose \textit{spacetime programming} (or ``spacetime'' for short) to tackle this compositionality issue.
Spacetime is a language based on the imperative synchronous language \textsf{Esterel}~\cite{esterel} and timed concurrent constraint programming (TCC)~\citep{tcc-lics94,saraswat-tcc-2014}.
Spacetime extends the synchronous model of computation of \textsf{Esterel} with backtracking, and refines the interprocess communication mechanism of TCC with lattice-based variables.
We introduce these features in the following two paragraphs.

\paragraph{Synchronous Programming with Backtracking}
The synchronous paradigm~\cite{synchronous-Halbwachs} proposes a notion of logical time dividing the execution of a program into a sequence of discrete instants.
A synchronous program is composed of processes that wait for one another before the end of each instant.
Operationally, we can view a synchronous program as a coroutine: a function that can be called multiple times and that maintains its state between two successive calls.
One call to this coroutine represents one instant that elapsed.
The main goal of logical time is to coordinate concurrent processes while avoiding typical issues of parallelism, such as deadlock or indeterminism~\cite{lee-threads}.

Spacetime inherits most of the temporal statements of TCC, and more specifically those of the synchronous language \textsf{Esterel}~\cite{esterel}, including the delay, sequence, parallel, loop and conditional statements.
The novelty of spacetime is to connect the search tree generated by a CSP and linear logical time of synchronous programming.
Our proposal is captured in the following principle:
\begin{center}
\textit{A node of the search tree is explored in exactly one logical instant.}
\end{center}
\noindent
A corollary to this first principle is:
\begin{center}
\textit{A search strategy is a synchronous process.}
\end{center}
\noindent
These two principles are illustrated in Sections~\ref{spacetime-programming} and~\ref{advanced-strategies} with well-known search strategies.

\paragraph{Deterministic Interprocess Communication}
The second characteristic of spacetime is inherited from concurrent constraint programming (CCP)~\cite{saraswat-semantic-1991}.
CCP defines a shared memory as a global constraint store accumulating partial information.
The CCP processes communicate and synchronize through this constraint store with two primitives: $tell(c)$ for adding a constraint $c$ into the store, and $ask(c)$ for asking if $c$ can be deduced from the store.
Concurrency is treated by requiring the store to grow \textit{monotonically} and \textit{extensively}, which implies that removal of information is not permitted.
An important result is that any CCP program is a closure operator over its constraint store (a function that is idempotent, extensive and monotone).

TCC embeds CCP in the synchronous paradigm~\cite{tcc-lics94,saraswat-tcc-2014} such that an instant is guaranteed to be a closure operator over its store; however information can be lost between two instants.
There are two main differences between spacetime and TCC.

Firstly, instead of a central and shared constraint store, variables in spacetime are defined over lattice structures.
The tell and ask operations are thus defined on lattices, where \textit{tell} relies on the least upper bound operation and \textit{ask} on the order of the lattice.
In Section~\ref{lattice-csp}, we formalize a CSP as a lattice that we later manipulate as a variable in spacetime programs.

Secondly, unlike TCC programs, spacetime programs are not closure operators by construction.
This stems from the negative ask statement (testing the absence of information) which is not monotone, and the presence of external functions which are not necessarily idempotent and monotone.
As in \textsf{Esterel}, we focus instead on proving that the computation is deterministic and reactive.
In addition, we also prove that spacetime programs are extensive functions \textit{within} and \textit{across} instants (Section~\ref{proof-correctness}).

\paragraph{Contributions}

In summary, this paper includes the following contributions:

\begin{itemize}
  \item We provide a language tackling the compositionality issue of search strategies.
        We illustrate this claim in Sections~\ref{spacetime-programming} and~\ref{advanced-strategies} by reconstructing and combining well-known search strategies.
  \item We extend the behavioral semantics of \textsf{Esterel} to backtracking and variables defined over lattices with proofs of determinism, reactivity and extensiveness (Section~\ref{spacetime-semantics}).
  \item We implement a prototype of the compiler\footnote{Open source compiler available at \url{https://github.com/ptal/bonsai/tree/PPDP19}.}, and integrate spacetime into the \textsf{Java} language (Section~\ref{implementation}).
        The evaluation of the search strategies presented in this paper shows a small overhead compared to the hard-coded ones of \textsf{Choco}~\cite{choco}.
  \item Spacetime is the first language that unifies constraint-based concurrency, synchronous programming and backtracking.
        This unification bridges a gap between the theoretical foundations of CCP and the practical aspects of constraint solving.
\end{itemize}

\section{Definitions}
\label{definitions}

To keep this paper self-contained, we expose necessary definitions on lattice theory which are then used to define constraint programming.
Given an ordered set $\langle L, \leq \rangle$ and $S \subseteq L$, $x \in L$ is a \textit{lower bound} of $S$ if $\forall{y \in S},~x \leq y$.
We denote the set of all the lower bounds of $S$ by $S^\ell$.
The element $x \in L$ is the \textit{greatest lower bound} of $S$ if $\forall{y \in S^\ell},~x \geq y$.
The \textit{least upper bound} is defined dually by reversing the order.

\begin{definition}[Lattice]
An ordered set $\langle L, \leq \rangle$ is a lattice if every pair of elements $x,y \in L$ has both a least upper bound and a greatest lower bound.
We write $x \sqcup y$ (called join) the least upper bound of the set $\{x, y\}$ and $x \sqcap y$ (called meet) its greatest lower bound.
A \textit{bounded lattice} has a top element $\top \in L$ such that $\forall{x \in L},~x \leq \top$ and a bottom element $\bot \in L$ such that $\forall{x \in L},\bot \leq x$.
\end{definition}
\noindent
As a matter of convenience and when no ambiguity arises, we simply write $L$ instead of $\langle L, \leq \rangle$ when referring to ordered structures.
Also, we refer to the ordering of the lattice $L$ as $\leq_L$ and similarly for any operation defined on $L$.

An example is the lattice $\mathit{LMax}$ of increasing integers $\langle N, \geq~,\mathit{max} \rangle$ where $N \subset \mathbb{N}$, $\geq$ is the natural order on $\mathbb{N}$ and $\mathit{max}$ is the join operator.
Dually, we also have $\mathit{LMin}$ with the order $\leq$ and join $\mathit{min}$.

The Cartesian product $P \times Q$ is defined by the lattice $\langle \{(x,y) \;|\; x \in P, y \in Q\}, \leq_\times\rangle$ such that $(x_1,y_1) \leq_\times (x_2, y_2)$ if $x_1 \leq_{P} x_2 \land y_1 \leq_{Q} y_2$.
Given the lattice $L_1 \times L_2$, it is useful to define the following projection functions, for $i \in \{1,2\}$ and $x_i \in L_i$ we have $\pi_i((x_1,x_2)) \mapsto x_i$.
For the sake of readability, we also extend the projection over any subset $S \subseteq L_1 \times L_2$ as $\pi'_i(S) = \{\pi_i(x) \;|\; x \in S\}$.

Given a lattice $\langle L, \leq\rangle$, a function $f: L \rightarrow L$ is extensive if for all $x \in L$, we have $x \leq f(x)$.
This property is important in language semantics because it guarantees that a program does not lose information.
More background on lattice theory can be found in~\cite{birkhoff-lattice-1967,davey-introduction-2002}.

\section{Lattice View of Constraint Programming}
\label{lattice-csp}

As we will see shortly, a spacetime program is a function exploring a state space defined over a lattice structure.
To illustrate this paradigm, we choose in this paper to focus on the state space generated by constraint satisfaction problems (CSPs).
Hence, we describe the lattice of CSPs and the lattice of its state space, called a \textit{search tree}.

\subsection{Lattice of CSPs}

Following various works~\cite{apt-essence-1999,Fernandez:2004:ICS:963778.963779,pelleau-constraint-2013,scott-other-2016}, we introduce constraint programming through the prism of lattice theory.
The main observation is that the hierarchical structure of constraint programming can be defined by a series of lifts.
We incrementally construct the lattice of CSPs.

First of all, we define the domain of a variable as an element of a lattice structure.
In the case of finite domains, an example is the powerset lattice $\langle \mathcal{P}(N), \supseteq \rangle$ with the finite set $N \subset \mathbb{N}$ and ordered by superset inclusion.
For instance, a variable $x$ in $\{0,1,2\} \in \mathcal{P}(N)$ is less informative than a singleton domain $\{0\}$, i.e. $\{0,1,2\} \leq \{0\}$.
Other lattices can be used (see e.g.~\cite{Fernandez:2004:ICS:963778.963779}), so we abstract the lattice of variable's domains as $\langle D, \leq \rangle$.

Let $\locset$ be an unordered set of variable's names.
We lift the lattice of domains $D$ to the lattice of partial functions $\locset \rightharpoonup D$.
In operational terms, a partial function represents a store of variables.
\begin{definition}[Store of variables]
We write the set of all partial functions from $\locset$ to $D$ as $[\locset \rightharpoonup D]$.
Let $\sigma, \tau \in [\locset \rightharpoonup D]$.
We write $\pi'_1(\sigma)$ the subset of $\locset$ on which $\sigma$ is defined.
The set of variables stores is a lattice defined as:
\begin{displaymath}
\mathit{SV}=\langle [\locset \rightharpoonup D], \tau \leq \sigma \text{ if } \forall{\ell \in \pi'_1(\tau)},~\tau(\ell) \leq_{D} \sigma(\ell) \rangle
\end{displaymath}
\label{var-store}
\end{definition}
\noindent
We find convenient to turn a partial function $\sigma$ into a set, called its graph, defined by $\{(x,\sigma(x))\;|\;x \in \pi'_1(\sigma)\}$.
Given a lattice $L$, the lattice $\storelat(\locset, L)$ is the set of the graphs of all partial functions from $\locset$ to $L$.
In comparison to $\mathit{SV}$, we parametrize the lattice $\storelat(\locset,L)$ by its set of locations $\locset$ and underlying lattice $L$, so we can reuse it later.
Notice that $\storelat(\locset,D)$ is isomorphic to $\mathit{SV}$.

We turn a logical constraint $c \in C$ into an extensive function $p: SV \to SV$, called \textit{propagator}, over the store of variables.
For example, given the store $d=\{x \mapsto \{1,2\}, y \mapsto \{2,3\}\}$ and the constraint $x \geq y$, a propagator $p_{\geq}$ associated to $\geq$ gives $p_{\geq}(d) = \{x \mapsto \{2\}, y \mapsto \{2,3\}\}$.
We notice that this propagation step is extensive, e.g. $d \leq p_{\geq}(d)$.
Beyond extensiveness, a propagator must also be sound, i.e. it does not remove solutions of the induced constraint, to guarantee the correctness of the solving algorithm.

We now define the lattice of all propagators $\mathit{SC}=\langle \mathcal{P}(\prop), \subseteq~\rangle$ where $\prop$ is the set of all propagators (extensive and sound functions).
The order is given by set inclusion: additional propagators bring more information to the CSP.
We call an element of this lattice a \textit{constraint store}.
The lattice of all CSPs---with propagators instead of logical constraints---is given by the Cartesian product $\csp=SV \times SC$.

Given a CSP $\langle d, \{p_1,\ldots,p_n\} \rangle \in \csp$, the \textit{propagation step} is realized by computing the fixpoint of $p_1(p_2(..p_n(d)))$.
We note $propagate: \csp \to \csp$ the function computing this fixpoint.
In practice, this function is one crucial ingredient to obtain good performance, and this is part of the theory of constraint propagation (e.g. see~\cite{apt-essence-1999,schulte-efficient-2008,propagation-guido-tack}).
In the rest of this paper, we keep this propagation step abstract, and we delegate it to specialized solvers when needed.

Once propagation is at a fixpoint, and if the domain $d$ is not a solution yet, a search step must be performed.
Search consists in splitting the state space with a branching function $\mathit{branch}: \csp \to \storelat(\mathbb{N}, \csp)$ and exploring successively the sub-problems created.
We call an element of the lattice $\storelat(\mathbb{N}, \csp)$ the \textit{branches}.
The indices of the branches serve to order the child nodes.
For instance, a standard branching function consists in selecting the first non-instantiated variable and to divide its domain into two halves---one explored in each sub-problem.
If the branching strategy is strictly extensive ($x < f(x)$) over each branch $b_i \in \mathit{branch}(\langle d, P\rangle)$, and does not add variables into $d$, then this solving procedure is guaranteed to terminate on finite domains.
This solving algorithm is called \textit{propagate and search}.

\subsection{Lattice of Search Trees}

A novel aspect of this lattice framework is to view the search tree as a lattice as well.
It relies on the \textit{antichain completion} which derives a lattice to the antichain subsets of its powerset.\footnote{In the finite case, the antichain completion of a lattice $L$ is isomorphic to the set of ideals of $L$ as shown by \citet{crampton-completion-2001}. We prefer the antichain formulation because it is closer to the data structure of a queue.}

\begin{definition}[Antichain completion]
The antichain completion of a lattice $L$, written $\antiL$, is a lattice defined as:
\begin{displaymath}
\begin{array}{l}
\antiL = \langle \{ S \subseteq \mathcal{P}(L) \;|\; \forall{x,y\in S},~x \leq y \implies x = y\}, \\
\phantom{\antiL = \langle } S \leq Q \text{ if } \forall{y \in Q},~\exists{x \in S},~x \leq_{L} y \rangle
\end{array}
\end{displaymath}
It is equipped with the Smyth order~\cite{smyth-power-1978}.
\end{definition}

The lattice of the search trees is defined as $\ST=\mathscr{A}(CSP)$.
Intuitively, an element $q \in \ST$ represents the frontier of the search tree being explored.
The antichain completion accurately models the fact that parents' nodes are not stored in $q$.
Operationally, we view $q$ as a queue of nodes\footnote{Despite the name, this terminology of ``queue'' does not imply a particular queueing strategy, i.e. the order in which the nodes are explored.}, which is central to backtracking algorithms.

The missing piece to build and explore the CSP state space is the \textit{queueing strategy} which allows us to pop and push nodes onto the queue.

\begin{definition}[Queueing strategy]
Let $L$ be a lattice and $\antiL$ be its antichain completion.
The pair of functions
\begin{displaymath}
\begin{array}{l}
\mathit{pop}: \antiL \to \antiL \times L\\
\mathit{push}: \antiL \times \storelat(\mathbb{N}, L) \rightarrow \antiL\\
\end{array}
\end{displaymath}
is a queueing strategy if, for any extensive function $f: \antiL \times L \rightarrow \antiL \times \storelat(\mathbb{N}, L)$, the function composition $\mathit{push} \circ f \circ \mathit{pop}$ is extensive over $\antiL$.
\label{queueing-property}
\end{definition}
\noindent
In the context of CSP solving, we have $L = \csp$ and $\antiL = \ST$.
As examples of queueing strategies, we have depth-first search (DFS), breadth-first search (BFS) and best-first search.

The state space of a CSP $\langle d, P \rangle$ is explored by computing the fixpoint of the function $\solve(\{\langle d, P \rangle\})$ which is defined as:
\begin{displaymath}
\begin{array}{l}
\solve: \ST \rightarrow \ST \\
\solve = \mathit{push} \circ (\mathit{id} \times (\mathit{branch} \circ \mathit{propagate})) \circ \mathit{pop}
\end{array}
\end{displaymath}
\noindent
This function formalizes the usual steps when solving a constraint problem: pop a node from the queue, propagate it, divide it into several sub-problems, and push these sub-problems onto the queue.
The output type of each function matches the input type of the next one---notice that we use the identity function $id$ to avoid passing the search tree to $\mathit{propagate}$ and $\mathit{branch}$.
Reaching a fixpoint on $\solve$ means that we explored the full search tree, and explored all solutions if there is any.

\subsection{The Issue of Compositionality}

The $\solve$ function is parametrized by a branching and queuing strategies.
However, this does not suffice to program every search strategy.
For example, the depth-bounded search strategy---further developed in the next section---consists in exploring the search tree until a given depth is reached.
To program this strategy in the current framework, we must extend the definition of a CSP with a depth counter defined over $\mathit{LMax}$ (given in Section~\ref{definitions}).
The resulting search tree is defined as $\ST_2 = \mathscr{A}(\csp \times \mathit{LMax})$.
We also extend $\solve$ with two functions: $\mathit{inc}$ for increasing the counter of the child nodes, and $\mathit{prune}$ for pruning the nodes at the given depth:
\begin{displaymath}
\begin{array}{l}
\mathit{solve2}: \ST_2\rightarrow \ST_2\\
\mathit{solve2} = \mathit{push} \circ (\mathit{id} \times (\mathit{inc} \circ \mathit{prune} \circ \mathit{branch} \circ \mathit{propagate} )) \circ \mathit{pop}
\end{array}
\end{displaymath}
Although orthogonal to the depth counter, the types of the $\mathit{propagate}$ and $\mathit{branch}$ functions must be modified to work over $\csp \times \mathit{LMax}$.
Another solution would be to project elements of $\csp \times \mathit{LMax}$ with additional $id$ functions.
A more elaborated version of this idea, relying on monads to encapsulate data, is investigated in \textit{monadic constraint programming}~\cite{mcp}.
The search strategies defined in this framework require the users to have substantial knowledge in functional language theory.
Similarly, constraint solving libraries are made extensible through software engineering techniques such as design patterns.
In all cases, a drawback is that it complicates the code base, which is hard to understand and extend with new search strategies.
Moreover, such software architecture varies substantially across solvers.

The problem is that we need to either \textit{modify existing structures} or integrate the strategies into some predefined software architecture in order to program new search strategies.
We call this problem the \textit{compositionality issue}.
Our proposal is to rely on \textit{language abstractions} instead of software abstractions to program search strategies.

\section{Language Overview}
\label{spacetime-programming}

We give a tour of the spacetime model of computation and syntax by incrementally building the iterative-deepening search strategy~\cite{Korf85depth-firstiterative-deepening}.
A key insight is that this search strategy is developed generically with regard to the state space.

\subsection{Model of Computation}

The model of computation of spacetime is inspired by those of (timed) concurrent constraint programming (CCP) and \textsf{Esterel}.

\paragraph{CCP model of computation}
\label{ccp}

We view the structure of a CCP program as a lattice $\langle L, \vDash, \sqcup \rangle$ where $\vDash$ is called the \textit{entailment}.
The entailment is the order of the lattice defined as $a \vDash b \equiv b \leq a$.
Following Scott's information systems~\cite{scott-domains-1982}, CCP views the bottom element $\bot$ as the lack of information, the top element $\top$ as all the information, the tell operator $x \sqcup y$ as the join of the information in $x$ and $y$, and the ask operator $x \vDash y$ as an expression that is true if we can deduce $y$ from $x$.

CCP processes communicate through this lattice by querying for information with the entailment, or adding information with join.
For example, consider the following definitions of \texttt{prune} and \texttt{inc}:
\begin{lstlisting}[mathescape,language=bonsai]
(when depth $\vDash$ 4 then $\text{``prune the subtree''}$) || (depth = depth $\sqcup$ (depth + 1))
\end{lstlisting}
\noindent
with $||$ the parallel composition.
The first process is suspended on \texttt{depth $\vDash$ 4} until \texttt{depth} becomes greater than or equal to $4$.
Hence, the second process is completed first if we initially have \texttt{depth < 4}.
The limitation of CCP is that it is not possible to write a process for the statement ``prune the subtree''.
This is because a CCP process computes over a fixed lattice, such as $\csp$, but it is not possible to compute over its antichain completion, which is necessary for creating and exploring its state space.

\paragraph{Space component of spacetime}

The approach envisioned with the spacetime paradigm is to view a search algorithm as a set of concurrent processes exploring collaboratively a state space.
In this model, we rewrite $solve2$ as a parallel composition of processes as follows (the arrows indicate read/write operations):
\begin{center}
\begin{tikzpicture}
\node at (3,1) {$\mathit{depth} \in \mathit{LMax}$};
\node at (5.5,1) {$\langle d, P\rangle \in \csp$};
\node[right] at (0,0) {$\mathit{solve2} = \mathit{push} \circ (\mathit{inc} \<hpar> \mathit{prune} \<hpar> \mathit{branch} \<hpar> \mathit{propagate}) \circ \mathit{pop}$};
\node at (4.7, -1) {$\mathit{branches} \in \storelat(\mathbb{N},\csp \times \mathit{LMax})$};


\draw[<->] (2.5,0.3) -- (2.5,0.75);
\draw[<->] (3.5,0.3) -- (2.7,0.75);

\draw[<->] (5,0.3) -- (5,0.75);
\draw[<->] (6,0.3) -- (5.2,0.75);

\draw[<->] (3.5,-0.3) -- (3.5,-0.75);
\draw[<->] (4.5,-0.3) -- (3.7,-0.75);

\end{tikzpicture}
\end{center}
\noindent
Firstly, we pop a node from the queue which contains the variables $depth$ and $\langle d, P \rangle$.
Then, similarly to CCP, the processes communicate by reading and writing into these variables.
The Cartesian product of the variables, called the \textit{space} of the program, is automatically synthesised by the spacetime semantics.
This is reflected in the type $\csp \times \mathit{LMax}$ of $\mathit{branches}$.
The processes only manipulate $\mathit{branches}$ through dedicated statements, namely \texttt{space} and \texttt{prune} (that we introduce below).

\paragraph{Time component of spacetime}

One remaining question is how to synchronize processes so that every process waits for each other before the next node is popped?
Our proposal is to rely on the notion of synchronous time of \textsf{Esterel}.
During each instant, a process is executed until it encounters a special statement called \texttt{pause}.\footnote{To ensure cooperative behavior among processes, the amount of work to perform during an instant must be bounded in time.}
Once \texttt{pause} is reached, the process waits for all other processes to be paused or terminated.
The next instant is then started.

The novelty in spacetime is to connect the passing of time to the expansion of the search tree.
Concretely, an instant consists in performing three consecutive steps: pop a node, execute the processes until they are all paused, and push the resulting branches onto the queue.
We repeat these steps until the queue is empty or all processes are terminated.

We now detail this model of computation through several examples, notably by programming the $inc$ and $prune$ processes.
We delay the presentation of $propagate$ and $branch$ to Section~\ref{advanced-strategies}.


\subsection{Binary Search Tree}
\label{binary-tree}

A spacetime program is a set of Java classes augmented with \textit{spacetime class fields} (prefixed by the \texttt{single_space}, \texttt{world_line} or \texttt{single_time} keywords) and \textit{processes} (prefixed by \texttt{proc} or \texttt{flow} keywords).
The type of a spacetime field or local variable is a Java class that implements a lattice interface providing the entailment and join operators.
A process does not return a value; it acts as a coroutine mutating the spacetime variables in each instant.
In contrast, Java method calls are viewed as atomic operations in a spacetime process.

One of the simplest process in spacetime is to generate an infinite binary search tree:
\begin{lstlisting}[language=bonsai]
class Tree {
  public proc binary =
    loop
      space nothing end;
      space nothing end;
      pause;
    end }
\end{lstlisting}
\noindent
This process generates a binary tree in which every node is empty; we will decorate these nodes with data later.
A branch is created with the statement \texttt{space $p$ end} where the process $p$ describes the differences between the current node and the child node.
In the example, the difference is given by \texttt{nothing} which is the \textit{empty process} terminating immediately without effect, thus all generated nodes will be the same.

In each instant, four actions are realized (we connect these actions to the model of computation in parenthesis):
\begin{enumerate}
\item A node is popped from the queue (function $\mathit{pop}$).
\item The process is executed until we reach a \texttt{pause} statement (process between $\mathit{pop}$ and $\mathit{push}$).
\item We retrieve the sequence of branches, duplicate the backtrackable state\footnote{The backtrackable state is the Cartesian product of the variables prefixed by \texttt{world_line} (see below).} for each \texttt{space $p$ end} statement, and execute each $p$ on a distinct copy of the state to obtain the child nodes (writing into the variable $\mathit{branches}$).
\item The child nodes are pushed onto the queue (function $\mathit{push}$).
\end{enumerate}
These actions are repeated in the statement \texttt{loop}.
Since the process \texttt{binary} never terminates and the queue is never empty, the state space generated is infinite.
In summary, a process generates a sequence of branches during an instant, and a search tree across instants.

Now, we illustrate the use of spacetime variables by introducing a node and depth counters:
\begin{lstlisting}[language=bonsai]
class Node {
  public single_space LMax node = new LMax(0);
  public flow count = readwrite node.inc() }
class Depth {
  public world_line LMax depth = new LMax(-1);
  public flow count = readwrite depth.inc() }
\end{lstlisting}
\noindent
A \textit{flow process} executes its body $p$ in each instant, the keyword \texttt{flow} is a syntactic sugar for \texttt{loop $p$; pause end}.
Both classes work similarly: we increase by one their counters in each instant with the method \texttt{inc} on \texttt{LMax}.
We discuss two kinds of annotations appearing in these examples: \textit{read/write annotations} and \textit{spacetime annotations}.

Read/write annotations indicate how a variable is manipulated inside a host function.
It comes in three flavors: \texttt{read $x$} indicates that $x$ is only read by the function, \texttt{write $x$} that the function only writes more information in $x$ without reading it, and \texttt{readwrite $x$} that the value written in $x$ depends on the initial value of $x$.
Every write in $x$ must respect its lattice order and this verification is left to the programmer of the lattice.
For example, the method \texttt{x.inc()} is defined as $x = x + 1$, and thus $x$ must be annotated by \texttt{readwrite}.
These attributes are essential to ensure determinism when variables are shared among processes, and for correctly scheduling processes.

\begin{figure*}[t!]
\begin{center}
\begin{tikzpicture}
\tikzstyle{time}=[draw,shape=circle,circle,fill,inner sep=1pt]
\tikzstyle{current}=[draw,shape=circle,circle,fill,inner sep=2pt]
\tikzstyle{explore}=[draw,shape=circle,circle,fill=white,inner sep=1.5pt]

\begin{scope}[yshift=4.2cm]
\node[anchor=south] at (3.75,0) {$t_1$};
\node[anchor=south] at (5.65,0) {$t_2$};
\node[anchor=south] at (7.9,0) {$t_3$};
\node[anchor=south] at (11.1, 0) {$t_6$};
\node[anchor=south] at (14.1325, 0) {$t_7$};
\end{scope}

\begin{scope}[xshift=0cm]

  \draw[rounded corners] (3.2,4.2) rectangle (4.2,2.8);

  \draw[thick] (4.2,3.5) node{} -- (4.5,3.5) node[time]{} -- (4.8,3.5);
  \draw[thick] (3,3.5) node[time]{} -- (3.2,3.5);

  \begin{scope}[xshift=0.7cm]
  \draw[thick] (3,3.7) node[current]{} -- (2.7,3.3) node[explore]{};
  \draw[thick] (3,3.7) node[current]{} -- (3.3,3.3) node[explore]{};
  \end{scope}
\begin{scope}[xshift=2.3cm]
  \draw[rounded corners] (2.5,4.2) rectangle (4.2,2.8);

  \draw[thick] (4.2,3.5) node{} -- (4.5,3.5) node[time]{} -- (4.8,3.5);

  \begin{scope}[xshift=0.5cm]

  \draw[thick] (3,4) node[time]{} -- (2.6,3.5) node[current]{};
  \draw[thick] (3,4) node[time]{} -- (3.4,3.5) node[explore]{};

  \draw[thick] (2.6,3.5) node[time]{} -- (2.3,3) node[explore]{};
  \draw[thick] (2.6,3.5) node[time]{} -- (2.9,3) node[explore]{};
  \end{scope}

\begin{scope}[xshift=2.3cm]
  \draw[rounded corners] (2.5,4.2) rectangle (4.2,2.8);
  \draw[thick] (4.2,3.5) -- (4.5,3.5) node[time]{};
  \draw[thick,densely dotted] (4.5,3.5) -- (5,3.5);
  \draw[thick] (5,3.5) node[time]{} -- (5.3,3.5);

  \begin{scope}[xshift=0.5cm]

  \draw[thick] (3,4) node[time]{} -- (2.6,3.5) node[time]{};
  \draw[thick] (3,4) node[time]{} -- (3.4,3.5) node[explore]{};

  \draw[thick] (2.6,3.5) node[time]{} -- (2.3,3) node[current]{};
  \draw[thick] (2.6,3.5) node[time]{} -- (2.9,3) node[explore]{};
  \end{scope}

\begin{scope}[xshift=3.5cm]
  \draw[rounded corners] (1.8,4.2) rectangle (4.2,2.8);
  \draw[thick] (4.2,3.5) node{} -- (4.5,3.5) node[time]{} -- (4.8,3.5);

  \draw[thick] (3,4) node[time]{} -- (2.4,3.5) node[time]{};
  \draw[thick] (3,4) node[time]{} -- (3.6,3.5) node[time]{};

  \draw[thick] (2.4,3.5) node[time]{} -- (2.7,3) node[time]{};
  \draw[thick] (2.4,3.5) node[time]{} -- (2.1,3) node[time]{};

  \draw[thick] (3.6,3.5) node[time]{} -- (3.3,3) node[current]{};
  \draw[thick] (3.6,3.5) node[time]{} -- (3.9,3) node[explore]{};

\begin{scope}[xshift=3cm]
  \draw[thick] (1.5,3.5) node[time]{} -- (1.8,3.5);
  \draw[rounded corners] (1.8,4.2) rectangle (4.2,2.8);

  \draw[thick] (3,4) node[time]{} -- (2.4,3.5) node[time]{};
  \draw[thick] (3,4) node[time]{} -- (3.6,3.5) node[time]{};

  \draw[thick] (2.4,3.5) node[time]{} -- (2.7,3) node[time]{};
  \draw[thick] (2.4,3.5) node[time]{} -- (2.1,3) node[time]{};

  \draw[thick] (3.6,3.5) node[time]{} -- (3.3,3) node[time]{};
  \draw[thick] (3.6,3.5) node[time]{} -- (3.9,3) node[current]{};
\end{scope}
\end{scope}
\end{scope}
\end{scope}
\end{scope}
\end{tikzpicture}
\end{center}
\caption{Progression of bounded depth search in each instant with maximum depth equals to $2$.}
\label{bds-fig}
\end{figure*}
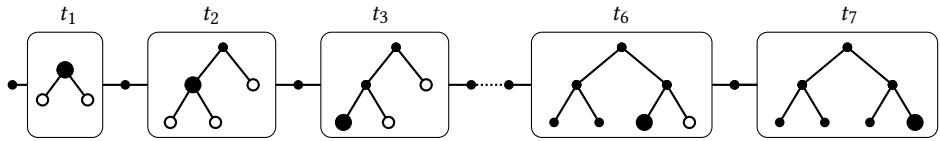
Spacetime annotations indicate how a variable evolves in memory through time.
For this purpose, a spacetime program has three distinct memories in which the variables can be stored:

\begin{enumerate}
\item[(i)] Global memory (keyword \texttt{single_space}) for variables evolving globally to the search tree.
A \texttt{single_space} variable has a unique location in memory throughout the execution.
For example, the counter \texttt{node} is a \texttt{single_space} variable: since we explore one node in every instant, we increase its value by one in each instant.
\item[(ii)] Backtrackable memory (keyword \texttt{world_line}) for variables local to a path in the search tree.
The queue of nodes is the backtrackable memory.
For example, the value of the counter \texttt{depth} must be restored on backtrack in the search tree.
\item[(iii)] Local memory (keyword \texttt{single_time}) for variables local to an instant and reallocated in each node.
A \texttt{single_time} variable only exists in one instant.
We will encounter this last annotation later on.
\end{enumerate}

Another feature of interest is the support of \textit{modular programming} by assembling processes defined in different classes.
As an example, we combine \texttt{Tree.binary} and \texttt{Depth.count} with the parallel statement:
\begin{lstlisting}[language=bonsai,mathescape]
public proc binary_stats =
  module Tree generator = new Tree();
  module Depth depth = new Depth();
  par run generator.binary()$\texttt{ || }$ run depth.count() end
end
\end{lstlisting}
\noindent
The variables \texttt{generator} and \texttt{depth} are annotated with \texttt{module} to distinguish them from spacetime variables.
We use the keyword \texttt{run} to disambiguate between process calls and method calls.

Last but not least, the \textit{disjunctive parallel} statement \texttt{par $p$ || $q$ end} executes two processes in lockstep.
It terminates once \textit{both} processes have terminated.
Dually, we have the \textit{conjunctive parallel} statement \texttt{par $p$ <> $q$ end} which terminates
\begin{enumerate*}
\item[(i)] in the next instant if one of $p$ or $q$ terminates, or
\item[(ii)] in the current instant if both $p$ and $q$ terminate.
\end{enumerate*}
The condition (i) implements a form of \textit{weak preemption}.
An instant terminates once every process is paused or terminated.
In this respect, \texttt{pause} can be seen as a synchronization barrier among processes.

\subsection{Depth-bounded Search}
\label{ids}

Now we are ready to program a search strategy in spacetime.
We consider the strategy \texttt{BoundedDepth} which bounds the exploration of the search tree to a depth limit:
\begin{lstlisting}[language=bonsai]
public class BoundedDepth {
  single_space LMax limit;
  public BoundedDepth(LMax limit) { this.limit = limit; }
  public proc bound_depth =
    module Depth counter = new Depth();
    par
    <> run counter.count()
    <> flow
        when counter.depth |= limit then prune end
       end
    end
  end }
\end{lstlisting}
\noindent
Whenever \texttt{depth} is greater than or equal to \texttt{limit} we prune the remaining search subtree.
The construction of the search tree through time is illustrated in Figure~\ref{bds-fig} with \texttt{limit} set at $2$.
The black dots are the nodes already visited, the large one is the one currently being visited and the white ones are those pushed onto the queue.

The disjunctive parallel composes two search trees by union, whereas the conjunctive parallel composes them by intersection.
For example, if we have \texttt{binary() || bound_depth()}, the search tree obtained is exactly the one of \texttt{binary()}, while \texttt{binary() <> bound_depth()} prunes the search tree at some depth limit.
Over two branches, the statement \texttt{prune || space $p$} creates a single branch \texttt{space $p$}, while \texttt{prune <> space $p$} creates a pruned branch.
This is made clear in Section~\ref{search-semantics} where we formalize these composition rules.

\subsection{A Glimpse of the Runtime}
\label{peek-runtime}

The class \texttt{Tree} is processed by the spacetime compiler which compiles every process into a regular Java method.
For example, the process \texttt{binary} is compiled into the following Java method:
\begin{lstlisting}[language=java]
public Statement binary() {
  return new Loop(
    new Sequence(Arrays.asList(
      new SpaceStmt(new Nothing()),
      new SpaceStmt(new Nothing()),
      new Delay(CompletionCode.PAUSE)))); }
\end{lstlisting}
\noindent
The compiled method returns the abstract syntax tree (AST) of the process.
This AST is then interpreted by the runtime engine \texttt{SpaceMachine}:
\begin{lstlisting}[mathescape,language=java]
public static void main(String[] args) {
  Tree tree = new Tree();
  StackLR queue = new StackLR();
  SpaceMachine machine = new SpaceMachine(tree.binary(), queue);
  machine.execute(); }
\end{lstlisting}
\noindent
We parametrize the runtime engine by the queue \texttt{StackLR}: a traditional stack exploring the tree in depth-first search from left to right.
Importantly, it means that the spacetime program is generic with regard to the queueing strategy.
The method \texttt{execute} returns either when the spacetime program terminates, the queue becomes empty or we reach a \texttt{stop} statement.
This latest statement offers a way to stop and resume a spacetime program outside of the spacetime world, which is handy for interacting with the external world.
In contrast, a \texttt{pause} statement is resumed automatically by the runtime engine as long as the queue is not empty.

Being aware of the runtime mechanism is helpful to extend \texttt{BoundedDepth} to the restart-based strategy \textit{iterative depth-first search} (IDS)~\cite{Korf85depth-firstiterative-deepening}.
IDS successively restarts the exploration of the same search tree by increasing the depth limit.
This strategy combines the advantages of breadth-first search (diversifying the search) and depth-first search (weak memory consumption).
Assuming we have a class \texttt{BoundedTree} combining \texttt{BoundedDepth} and \texttt{Tree}, we program IDS in the host language as follows:
\begin{lstlisting}[language=java]
public static void main(String[] args) {
  for(int limit=0; limit < max_depth(); limit++) {
    BoundedTree tree = new BoundedTree(new LMax(limit));
    StackLR queue = new StackLR();
    SpaceMachine machine = new SpaceMachine(tree.search(), queue);
    machine.execute(); }}
\end{lstlisting}
\noindent
We introduce additional examples of search strategies in Section~\ref{advanced-strategies}, and show how to combine two restart-based strategies in spacetime.

\section{Semantics of Spacetime}
\label{spacetime-semantics}

We develop the semantics of spacetime independently from the host language (\textsf{Java} in the previous section).
To achieve that, we suppose the program is flattened: every module definition and process call are inlined, and no recursion is allowed in processes.
We obtain a lighter abstract syntax of the spacetime statements formalized as follows ($p,q$ are processes, $x,y$ are identifiers, and $T$ is a host type):

\setlength{\grammarindent}{4em}
\begin{grammar}
<p, q> ::= \texttt{$T$ $x^{\rightarrow|\circlearrowleft|\downarrow}$} $\;$|$\;$ \texttt{when \texttt{$x$ |= $y$} then $p$ else $q$} 
\alt $f$($x_1^{w|r|\rw},\ldots,x_n^{w|r|\rw}$)
\alt \texttt{nothing} | \texttt{pause} | \texttt{stop} | \texttt{loop $p$} | $p \<seq> q$ | $p \<hpar> q$ | $p \<dia> q$ 
\alt \texttt{space} $p$ | \texttt{prune} 
\end{grammar}
\noindent
Spacetime annotations are shorten as follows: $\rightarrow$ stands for \texttt{single_space}, $\circlearrowleft$ for \texttt{single_time} and $\downarrow$ for \texttt{world_line}.\footnote{These symbols reflect how the variables evolve in the search tree. For example, $\downarrow$ depicts an evolution from the root to a leaf of the tree along a path.}
Read/write annotations are given by $w$ for \texttt{write}, $r$ for \texttt{read} and $\rw$ for \texttt{readwrite}.
Without loss of generality, we encapsulate the interactions between spacetime and its host language in function calls.


\subsection{Behavioral Semantics}
\label{behavioral-semantics}

The semantics of spacetime is inspired by the logical behavioral semantics of \textsf{Esterel}, a big-step semantics, as defined in~\cite{berry-constructive-2002,esterel-compilation}.
The semantic rules of spacetime defining the control flow of processes (for example \texttt{loop} or \texttt{pause}) are similar to those in \textsf{Esterel}.
We adapt these rules to match the two novel aspects of spacetime:
\begin{enumerate}
\item[(i)] Storing lattice-based variables in one of the three memories (instead of Esterel's Boolean signals).
\item[(ii)] Defining a structure to collect and compose the (pruned) branches created during an instant.
\end{enumerate}
The rules proper to spacetime are specific to either (i) or (ii).

Given the set of outputs produced by a program, a derivation in the behavioral semantics is a proof that a program transition is valid.
The behavioral transition rule is given as:
\begin{displaymath}
Q, L \vdash p \xrightarrow[I \sqcup O]{O'} p'
\end{displaymath}
\noindent
where the program $p$ is rewritten into the program $p'$ under
\begin{enumerate*}
\item[(i)] the queue $Q$ equipped with a queueing strategy $(pop, push)$,
\item[(ii)] the set of locations $L \subset Loc$ providing a unique identifier to every declaration of variable,
\item[(iii)] the input $I$, and
\item[(iv)] the outputs $O$ and $O'$.
\end{enumerate*}
We denote the set of syntactic variable names (as appearing in the source code) with $\nameset$, such that $\nameset \cap Loc = \emptyset$.
We write $L \dcup \{\ell\}$ the disjoint union, which is useful to extract a fresh location $\ell$ from $L$.

The goal of behavioral semantics is not to compute an output $O$ but to prove that a transition is valid if we already know $O$.
We obtain a valid derivation if the output $O'$ derived by the semantics is equal to the provided output $O$.
Conceptually, the behavioral semantics allows processes to instantaneously broadcast information.
In the following, we call the input and output structures \textit{universe} and we write $U'$ for the output $O'$, and $U=I\sqcup O$ for the input/output provided.

\subsubsection*{Space structure}
\label{space-structure}

The variable environment of a program, called its \textit{space}, stores the spacetime variables.
The spacetime annotations are given by the set $\mathit{spacetime} = \{\rightarrow,\circlearrowleft,\downarrow\}$.
The set of values of a variable is given by its type in the host language, which must be a lattice structure.
From the spacetime perspective, we erase the types in the set $\mathit{Value}$ which is the disjoint union of all types, and we delegate typing issues to the host language.
Putting all the pieces together, the set of spacetime variables $\mathit{Var}$ is the poset $\{\top\} \cup (\mathit{spacetime} \times \mathit{Value})$.
We need a distinct top element $\top$ for representing variables that are merged with a different spacetime or type---this can be checked at compile-time.

Given a set of locations $\mathit{Loc}$, the lattice of the spaces of the program is defined as $\mathit{Space} = \mathit{Store}(\mathit{Loc}, \mathit{Var})$.
The element $\bot$ is the empty space.
Given a space $S \in \mathit{Space}$, we define the subsets of the single space variables with $S^\rightarrow$, the single time variables with $S^\circlearrowleft$ and the world line variables with $S^\downarrow$.
In addition, given a variable $(st, v) \in S(\ell)$ at location $\ell$, we define the projections $S^{\spt}(\ell) = \spt$ and $S^{V}(\ell) = v$ to respectively extract the spacetime and the value of the variable.
$S^{V}(\ell)$ maps to $\bot$ if $\ell$ is undefined in $S$.

\subsubsection*{Universe structure}
\label{universe}

A universe incorporates all the information produced during an instant including the space, the completion code and the sequence of branches.
The completion code models the state of a process at the end of an instant: normally terminated (code $0$), paused in the current instant with \texttt{pause} (code $1$) or stopped in the user environment with \texttt{stop} (code $2$).
We denote the set of completion codes with $\mathit{Compl} = \langle \{0,1,2\}, \leq_{\mathbb{N}}\rangle$.
We describe the sequence of branches $B^{*}$ in the next section.
The universe structure is defined as follows:
\begin{displaymath}
  \mathit{Universe} = \mathit{Space} \times \mathit{Compl} \times B^{*}
\end{displaymath}
Given $U \in \mathit{Universe}$, we define the projections $U^S$, $U^k$ and $U^B$ respectively mapping to the space, completion code and the sequence of branches.
We also write $U^V$ instead of $U^{S^V}$, $U^\rightarrow$ instead of $U^{S^\rightarrow}$ and similarly for $\circlearrowleft$ and $\downarrow$.

\newcommand{\sizerules}{\small}
\begin{figure*}
\begin{mathpar}
  \sizerules

\inferrule[\sizerules nothing]
{}
{Q, \{\} \vdash \texttt{nothing} \xrightarrow[U]{\bot,\,0\,\langle \rangle} \texttt{nothing}}

\inferrule[\sizerules pause]
{}
{Q, \{\} \vdash \texttt{pause} \xrightarrow[U]{\bot,\,1,\,\langle \rangle} \texttt{nothing}}

\inferrule[\sizerules stop]
{}
{Q, \{\} \vdash \texttt{stop} \xrightarrow[U]{\bot,\,2,\,\langle \rangle} \texttt{nothing}}

\inferrule[\sizerules hcall]
{f(\ell_1^{a_1},\ldots,\ell_n^{a_n}) \xrightarrowdbl[\mathit{host}(U^S)]{H'} v}
{Q, \{\} \vdash f(\ell_1^{a_1},\ldots,\ell_n^{a_n}) \;\xrightarrow[U]{(\mathit{space}(H'), 0, \langle \rangle)} \texttt{nothing}}

\inferrule[\sizerules loop]
{Q, L \vdash p \;\xrightarrow[U]{U'}\; p' \\ U^k \neq 0}
{Q, L \vdash \texttt{loop $p$} \xrightarrow[U]{U'}\; \texttt{$p'$ ; loop $p$}}

\inferrule[\sizerules when-true]
{U^V(\ell_1) \vDash U^V(\ell_2) \twoheadrightarrow \mathit{true} \\ Q, L \vdash p \;\xrightarrow[U]{U'}\; p'}
{Q, L \vdash \texttt{when $\ell_1$ |= $\ell_2$ then $p$ else $q$} \xrightarrow[U]{U'}\; p'}

\inferrule[\sizerules when-false]
{U^V(\ell_1) \vDash U^V(\ell_2) \twoheadrightarrow v \quad v = \mathit{false} \lor v = \mathit{unknown} \\ Q, L \vdash q \;\xrightarrow[U]{U'}\; q'}
{Q, L \vdash \texttt{when $\ell_1$ |= $\ell_2$ then $p$ else $q$} \xrightarrow[U]{U'}\; q'}

\inferrule[\sizerules var-decl$\circlearrowleft$]
{U' = (\{(\ell,(\circlearrowleft,\bot_{T}))\}, 0, \langle \rangle) \\
Q, L \vdash p[x \rightarrow \ell] \;\xrightarrow[U]{U''}\; p'}
{Q, L \dcup \{\ell\} \vdash \texttt{$T$\;$x^\circlearrowleft$ ; $p$} \;\xrightarrow[U]{U' \sqcup U''} \texttt{$T$\;$x^\circlearrowleft$ ; $p'$}}

\inferrule[\sizerules start-var-decl$\rightarrow\downarrow$]
{\spt \neq \circlearrowleft \\
x \in \nameset \\
U' = (\{(\ell,(\spt,\bot_{T}))\}, 0, \langle \rangle) \\
Q, L \vdash p[x \rightarrow \ell] \;\xrightarrow[U]{U''}\; p'}
{Q, L \dcup \{\ell\} \vdash \texttt{$T$\;$x^{\spt}$ ; $p$} \;\xrightarrow[U]{U' \sqcup U''} \texttt{$T$\;$\ell^{\spt}$ ; $p'$}}

\inferrule[\sizerules prune]
{}
{Q, \{\} \vdash \texttt{prune} \xrightarrow[U]{(\bot,0,\langle \texttt{prune} \rangle)} \texttt{nothing}}

\inferrule[\sizerules resume-var-decl$\rightarrow\downarrow$]
{\ell \in \locset \\
v = \left\{
{\begin{tabular}{ll}
\text{$(\rightarrow, \bot_T)$} & \text{$\text{ if } \spt = \rightarrow$} \\
\text{$(\downarrow, \pi_2(\mathit{pop}(Q))(\ell))$} & \text{$\text{ if } \spt = \downarrow$}
\end{tabular}} \right. \\
U' = (\{(\ell,v)\}, 0, \langle \rangle) \\
Q, L \vdash p \;\xrightarrow[U]{U''}\; p'}
{Q, L \vdash \texttt{$T$\;$\ell^{\spt}$ ; $p$} \;\xrightarrow[U]{U' \sqcup\, U''} \texttt{$T$\;$\ell^{\spt}$ ; $p'$}}

\inferrule[\sizerules space-pruned]
{U^{B} \neq \langle\texttt{space $W$}\rangle}
{Q, \{\} \vdash \texttt{space $p$} \xrightarrow[U]{(\bot,0,\langle \texttt{space $\bot$} \rangle)} \texttt{nothing}}

\inferrule[\sizerules space]
{U^{B} = \langle\texttt{space $W$}\rangle \\
\bot, \{\} \vdash p \;\xrightarrow[U \sqcup (W,0,\langle \rangle)]{U'}\; p' \\
 U'^k = 0 \\
 U'^\rightarrow = U'^\circlearrowleft = \emptyset}
{Q, \{\} \vdash \texttt{space $p$} \xrightarrow[U]{(\bot,0,\langle \texttt{space $U'^\downarrow$}\rangle)} \texttt{nothing}}

\inferrule[\sizerules enter-seq]
{Q, L \vdash p \xrightarrow[U]{U'} p' \\ U'^k \neq 0}
{Q, L \vdash \texttt{$p$ ; $q$} \xrightarrow[U]{U'} \texttt{$p'$ ; $q$}}

\inferrule[\sizerules next-seq]
{U^B = B \circ B' \\
Q, L \vdash p \xrightarrow[(U^S, U^k, B)]{U'} p' \\
U'^k = 0 \\
Q, L' \vdash q \xrightarrow[(U^S, U^k, B')]{U''} q'}
{Q, L \dcup L' \vdash \texttt{$p$ ; $q$} \xrightarrow[U]{U' \sqcup^\circ U''} q'}

\inferrule[\sizerules par$^\lor$]
{Q, L \vdash p \xrightarrow[U]{U'} p' \\
Q, L' \vdash q \xrightarrow[U]{U''} q'}
{Q, L \dcup L' \vdash \texttt{$p$ || $q$} \xrightarrow[U]{U' \sqcup^{\lor} U''} \texttt{$p'$ || $q'$}}

\inferrule[\sizerules par$^\land$]
{Q, L \vdash p \xrightarrow[U]{U'} p' \\
Q, L' \vdash q \xrightarrow[U]{U''} q' \\
U'^k \neq 0 \land U''^k \neq 0}
{Q, L \dcup L' \vdash \texttt{$p$ <> $q$} \xrightarrow[U]{U' \sqcup^{\land} U''} \texttt{$p'$ <> $q'$}}

\inferrule[\sizerules exit-par$^\land$]
{Q, L \vdash p \xrightarrow[U]{U'} p' \\
Q, L' \vdash q \xrightarrow[U]{U''} q' \\
U'^k = 0 \lor U''^k = 0}
{Q, L \dcup L' \vdash \texttt{$p$ <> $q$} \xrightarrow[U]{U' \sqcup^{\land} U''} \texttt{nothing}}

\end{mathpar}
\caption{Behavioral semantics rules of spacetime.}
\label{spacetime-semantics-fig}
\end{figure*}

\subsection{Search Semantics}
\label{search-semantics}

In this section, we use the following relevant subset of spacetime:
\begin{grammar}
<p, q> ::= $p \<seq> q$ | $p \<hpar> q$ | $p \<dia> q$ | \texttt{space} $p$ | \texttt{prune} | $\alpha$
\end{grammar}
where $p,q \in Proc$ with $Proc$ the set of all the processes, and $\alpha$ is an atomic statement which is not composed of other statements.
We can extend the definitions given below to the full spacetime language without compositional issues.

We give the semantics of the search tree statements with a branch algebra.
We have a set of all branches defined as $B = \{\texttt{space } w \;|\; w \in Space^\downarrow\} \cup \{\texttt{prune}\}$.
That is to say, a branch is either labelled by a \texttt{world_line} space or pruned.

\begin{definition}[Branch algebra]
The branch algebra is defined over a sequence of branches $\langle B^{*}, \circ, \lor, \land \rangle$ where all operators are associative, $\circ$ is noncommutative, and $\lor$ and $\land$ are commutative.
The empty sequence $\langle \rangle$ is the identity element of the three operators.
\end{definition}
\noindent
The operators $\circ$, $\lor$ and $\land$ match the commutative and associative laws of the semantics of the operators \texttt{;},\texttt{||} and \texttt{<>} respectively.

\paragraph{Sequence composition}
Given $b_i, b_j \in B$ with $1 \leq i \leq n$ and $1 \leq j \leq m$, the sequence operator $\circ$ performs the concatenation of two sequences of branches as follows:
\begin{displaymath}
\langle b_1, \dots, b_n \rangle \circ \langle b'_1, \dots, b'_m \rangle  = \langle \;b_1,\ldots,b_n,b'_1,\ldots,b'_m\;\rangle
\end{displaymath}

\paragraph{Parallel compositions}

We define the operators $\lor^{1}$ and $\land^{1}$ to combine two branches and then lift these operators to sequences of branches.
Two sequences of branches are combined by repeating the last element of the shortest sequence when the sizes differ.
Given $w, w' \in Space^\downarrow$ and $b \in B$, we define the disjunctive parallel operators $\lor^{1}$ between two branches and $\lor$ between two sequences of branches as follows:
\begin{displaymath}
\begin{array}{ll}
b \lor^{1} \texttt{prune} &= b \\
\texttt{space $w$} \lor^{1} \texttt{space $w'$} &= \texttt{space $w \sqcup w'$} \\
\langle b_1, \dots, b_n \rangle \lor \langle b'_1, \dots, b'_m \rangle  &=  \\
  \multicolumn{2}{l}{\quad \left \{
  \begin{array}{l}
  \langle \;b_1 \lor^{1} b'_1,\;\ldots,\;b_{n-1} \lor^{1} b'_{m-1},\;b_n \lor^{1} b'_m\;\rangle\;\text{ if $n = m$}\\
  \langle \;b_1 \lor^{1} b'_1,\;\ldots,\;b_{n-1} \lor^{1} b'_m,\;b_n \lor^{1} b'_m\;\rangle\;\text{ if $n > m$}
  \end{array} \right.}
\end{array}
\end{displaymath}
\noindent
The case where $m > n$ is tackled by the commutativity of $\lor$.
The conjunctive parallel operators $\land^{1}$ and $\land$ are defined similarly but for \texttt{prune}:
\begin{displaymath}
\begin{array}{l}
b \land^{1} \texttt{prune} = \texttt{prune}\\
\end{array}
\end{displaymath}

This algebra allows us to delete, replace or increase the information in a branch.
For example, given a process $p$:
\begin{itemize}
\item \texttt{$p$ <> (space nothing ; prune)} deletes every branch created by $p$ but the first.
\item \texttt{$p$ <> (space nothing ; prune ; space nothing)} deletes the second branch.
\item \texttt{$p$ || (prune ; space $q$ ; prune)} increases the information in the second branch by $q$.
\end{itemize}
We can also obtain any permutation of a sequence of branches with a suited $push$ function.
The only operation not supported is weakening the information of one branch.
We have yet to find a use-case for such an operation.

\begin{figure*}
\begin{center}
\newcommand{\examplesize}{\small}
\begin{mathpar}
\examplesize
\inferrule*[Left={\examplesize start-var-decl$\rightarrow\downarrow$}]
{
  \inferrule*[Left={\examplesize exit-par-$\lor$}]
  {
    \inferrule*[Left={\examplesize when-true}]
    {
      U^V(\ell_0) \vDash 1 \twoheadrightarrow \mathit{true} \\
      \inferrule*[Left={\examplesize space}]
      {
        \inferrule*[Left={\examplesize hcall}]
        {\mathit{inc}(\ell_0^{\rw}) \xrightarrowdbl[\mathit{host}(S_2)]{H'} v}
        {Q, \{\} \vdash \texttt{inc($\ell_0^{\rw}$)} \xrightarrow[(S_2, 0, \langle \rangle)]{(\mathit{space}(H'),0,\langle \rangle)} \texttt{nothing}}
      }
      {Q, \{\} \vdash \texttt{space inc($\ell_0^{\rw}$)} \xrightarrow[(S_1, 0, \langle \texttt{space $S_2$} \rangle)]{(\{\},0,\langle \texttt{space $S_2$} \rangle)} \texttt{nothing}}
    }
    {Q, \{\} \vdash \texttt{when $\ell_0$ |= 1 then space inc($\ell_0^{\rw}$)}\;\xrightarrow[(S_1, 0, \langle \texttt{space $S_2$} \rangle)]{(\{\}, 0, \langle \texttt{space $S_2$} \rangle)} \texttt{nothing}}
    \quad
    \inferrule[{\examplesize hcall}]
    {\mathit{inc}(\ell_0^{\rw}) \xrightarrowdbl[\mathit{host}(S_1)]{H'} v}
    {Q, \{\} \vdash \texttt{inc($\ell_0^{\rw}$)} \xrightarrow[(S_1, 0, \langle \texttt{space $S_2$} \rangle)]{(\mathit{space}(H'),0,\langle \rangle)} \texttt{nothing}}
  }
  {Q, \{\} \vdash \texttt{(when $\ell_0$ |= 1 then space inc($\ell_0^{\rw}$)) <> inc($\ell_0^{\rw}$)}\;\xrightarrow[(S_1, 0, \langle \texttt{space $S_2$} \rangle)]{(S_1, 0, \langle \texttt{space $S_2$} \rangle)} \texttt{nothing} \\ U' = (\{(\ell_0,(\downarrow, 0))\}, 0, \langle \rangle)}
}
{Q, \{\ell_0\} \vdash \texttt{LMax x$^\downarrow$; ((when x |= 1 then space inc(x$^{\rw}$)) <> inc(x$^{\rw}$))}\;\xrightarrow[(\{(\ell_0,(\downarrow, 1))\}, 0, \langle \texttt{space $\{(\ell_0,(\downarrow, 2))\}$} \rangle)]{U' \,\sqcup\, (\{(\ell_0,(\downarrow, 1))\}, 0, \langle \texttt{space $\{(\ell_0,(\downarrow, 2))\}$} \rangle)} \texttt{nothing}}
\end{mathpar}
\end{center}
\caption{An example of derivation in the behavioral semantics.}
\label{example-derivation}
\vspace{-0.1cm}
\end{figure*}

\subsection{Semantics Rules}
\label{semantics-rules-sec}

The semantics rules of spacetime are given in Figure~\ref{spacetime-semantics-fig}.
We isolate host computations by relying on the \textit{host transition rule} $e \xrightarrowdbl[~H~]{H'} v$ which reduces the expression $e$ into the value $v$ with the input/output host environment $H$ and the output environment $H'$.
The interface between spacetime and the host language is realized by a pair of functions $(\mathit{host}, \mathit{space})$ such that $\mathit{host}$ maps the space $S$ into the host environment $H$ and vice versa.
We write $e \twoheadrightarrow v$ when the space of the program is not modified.
We explain each fragment of the semantics in the following paragraphs.

The axioms \textsc{nothing}, \textsc{pause} and \textsc{stop} set the completion code respectively to terminated, paused and stopped.
We leave the output space and branches empty.

The main interaction with the host language is given by the rule \textsc{hcall}.
The function $f$ and its arguments are evaluated in the host version of the input/output space, written $\mathit{host}(U^S)$.
The properties guaranteed by the spacetime semantics depend on the properties fulfilled by the host functions.

The rule \textsc{loop} simulates an iteration of the loop by extracting and executing the body $p$ outside of the loop.
We guarantee that $p$ is not instantaneous by forbidding the completion code $k$ to be equal to $0$.

The conditional rules \textsc{when-true} and \textsc{when-false} evaluate the entailment result of $x \vDash y$ to execute either $p$ or $q$.
In case the entailment status is $\mathit{unknown}$, which happens if $x$ and $y$ are not ordered, we promote $\mathit{unknown}$ to $\mathit{false}$.
This is reminiscent of the closed world assumption in logic programming: ``what we do not know is $\mathit{false}$''.

\subsubsection{Semantics of spacetime variables}
\label{hierarchical-variable}

The variable declaration rules register the variables in the space or queue memory.
A variable's name $x$ must be substituted to a unique location $\ell$.
Locations are necessary to distinguish variables with the same name in the space and queue---this is possible if the scope of the variable is re-entered several times during\footnote{This is a problem known as reincarnation in \textsf{Esterel}~\cite{berry-constructive-2002}.} and across instants.
In the rules \textsc{var-decl$\circlearrowleft$} and \textsc{start-var-decl$\rightarrow\downarrow$}, we extract a fresh location $\ell$ from $L$ and substitute $x$ for $\ell$ in the program $p$, which is written $p[x \to \ell]$.\footnote{The variable declaration must be evaluated with regard to its body, this is why the body $p$ follows the declaration.
We can transform any variable declaration \texttt{$\mathit{Type}$ $x^{\spt}$} which is not followed by any statement to \texttt{$\mathit{Type}$ $x^{\spt}$; nothing}.}
The substitution function is defined inductively over the structure of the program $p$.
We give its two most important rules:
\begin{displaymath}
\begin{array}{l}
y[x \to \ell] \mapsto \left\{
  \begin{array}{ll}
    \ell & \text{ if } x = y \\
    y & \text{ if } x \neq y \\
  \end{array}\right.\\
(T\; y^{\spt} \<seq> p)[x \to \ell] \mapsto \left\{
  \begin{array}{ll}
    T\; y^{\spt} \<seq> p & \text{ if } x = y \\
    T\; y^{\spt} \<seq> p[x \to \ell] & \text{ if } x \neq y \\
  \end{array}\right.
\end{array}
\end{displaymath}
\noindent
It replaces any identifier equals to $x$ by $\ell$, and stops when it reaches a variable declaration with the same name.

For \texttt{single_time} variables, we create a new location in each instant (\textsc{var-decl$\circlearrowleft$}).
For \texttt{single_space} and \texttt{world_line} variables, we create a new location only during the first instant of the statement (\textsc{start-var-decl$\rightarrow\downarrow$}), and the next instants reuse the same location (\textsc{resume-var-decl$\rightarrow\downarrow$}).

In the first instant, the values are initialized to the bottom element $\bot_T$ of the lattice $T$.
In the next instants, we retrieve the value of a \texttt{world_line} variable in the queue by popping one node, and then extracting the value at location $\ell$ from that node.
The values of \texttt{single_space} variables are transferred from one instant to the next by the reaction rules introduced in the next section.

\subsubsection{Semantics of search statements}

The statement \texttt{prune} is an axiom creating a single pruned branch.
For \texttt{space $p$}, we have two cases: either we execute $p$ under the input/output branch $\langle \texttt{space $W$} \rangle$ (rule~\textsc{space}), or if another process prunes this branch, we avoid executing $p$ (rule~\textsc{space-pruned}).
The execution of the \texttt{space} statement does not impact the variables in the current instant, which is materialized by setting the space to $\bot$ in the output universe.
In addition, we require that $p$ terminates instantaneously, only writes into \texttt{world_line} variables and does not create nested branches.

To specify the sequential and parallel statements, we extend join over $\mathit{Universe}$ with a branch operator.
We have $(S, k, B)\sqcup^\land (S', k', B')$ equals to $(S \sqcup S', k \sqcup k', B \land B')$, and similarly for $\circ$ and $\lor$.

To formalize the sequence \texttt{$p$ ; $q$}, we have the rule \textsc{enter-seq} which tackles the case where $p$ does not terminate during the current instant, and the rule \textsc{next-seq} where $p$ terminates and $q$ is executed.
The disjunctive parallel statement \texttt{$p$ || $q$} derives $p$ and $q$ concurrently and merges their output universes with $\sqcup^\lor$ (rule \textsc{par$^\lor$}).
Finally, the conjunctive parallel statement \texttt{$p$ <> $q$} is similar to \texttt{||} when none of $p$ or $q$ terminates (rule \textsc{par$^\land$}).
However, if one process terminates, we rewrite the statement to \texttt{nothing} which prevents this statement to be executed in future instants (rule \textsc{exit-par$^\land$}).
Note that the semantics of composition in space of \texttt{||} and \texttt{<>} match their respective semantics of composition in time.

\subsubsection{An example of derivation}

We illustrate the mechanics of the behavioral semantics with a short example:
\begin{lstlisting}[mathescape,language=bonsai]
LMax x$^\downarrow$; ((when x |= 1 then space inc(x$^{\rw}$)) $\texttt{<>}$ inc(x$^{\rw}$))
\end{lstlisting}
Two processes communicate over the variable $x$.
The first creates a branch incrementing $x$ by one if it is greater than $1$, while the second increments $x$ in the current instant.
To derive this process in the behavioral semantics, we set the input/output universe to $U = (\{(\ell_0,(\downarrow, 1))\}, 0, \langle \texttt{space $\{(\ell_0,(\downarrow, 2))\}$} \rangle)$ and attempt to prove that the output universe (the structure above the arrow) is equal to $U$.
For clarity, we set $S_1 = \{(\ell_0,(\downarrow, 1))\}$ and $S_2 = \{(\ell_0,(\downarrow, 2))\}$.
The derivation is given in Figure~\ref{example-derivation}.
We notice that the statement \texttt{space} is derived with the input/output space $S_2$ instead of $S_1$.
Operationally, it implies that the branch must be evaluated at the end of the current instant.

\subsection{Semantics Across Instants}
\label{universe-semantics}

\begin{figure}
\begin{mathpar}
\newcommand{\sizerulesa}{\small}
\sizerulesa
  \inferrule[\sizerulesa react]
  {\mathit{causal}(p) \\
   Q, \mathcal{L}_i \vdash p \xrightarrow[H_i]{U'} p' \\
   Q' = \mathit{push}(Q, U'^B) \\
   U'^k = 1 \text{ and $Q'$ is not empty} \\
   i+1, \mathcal{L} \vdash \langle Q', p' \rangle \xhookrightarrow[H]{H'} \langle Q'', p'' \rangle \\
   H'' = \{(j,U''\sqcup(U'^\rightarrow,0,\langle\rangle)) \;|\; (j,U'') \in H'\}}
  {i, \mathcal{L} \vdash \langle Q, p \rangle \xhookrightarrow[H]{\{(i,U')\} \sqcup H''} \langle Q'', p'' \rangle}

  \inferrule[\sizerulesa exit-react]
  {\mathit{causal}(p) \\
   Q, \mathcal{L}_i \vdash p \xrightarrow[H_i]{U'} p' \\
   Q' = \mathit{push}(Q, U'^B)\\
   U'^k \neq 1 \text{ or $Q'$ is empty}}
  {i, \mathcal{L} \vdash \langle Q, p \rangle \xhookrightarrow[H]{\{(i,U')\}} \langle Q', p' \rangle}
\end{mathpar}
\caption{Reaction rules of spacetime.}
\label{react-rules}
\end{figure}

A spacetime program is automatically executed until it terminates, stops or its queue of nodes becomes empty.
Therefore, we must lift the transition rule to succession of instants, which gives the following \textit{reaction rule}:
\begin{displaymath}
i, \mathcal{L} \vdash \langle Q, p \rangle \xhookrightarrow[H]{H'} \langle Q', p' \rangle
\end{displaymath}
where the state $\langle Q, p \rangle$ is rewritten into the state $\langle Q', p' \rangle$ with $Q$ a queue with a queueing strategy $(\mathit{pop},\mathit{push})$, and $p$ a process.
In addition, we have:
\begin{enumerate*}
\item[(i)] a counter of instants $i \in \mathbb{N}$,
\item[(ii)] a sequence of sets of locations $\mathcal{L} \in \storelat(\mathbb{N}, \locset)$ where $\mathcal{L}_i \in \mathcal{L}$ is the set of locations at the instant $i$,
\item[(iii)] the sequence of input/output universes $H \in \storelat(\mathbb{N}, \mathit{Universe})$ where $H_i$ is the input/output at the instant $i$, and
\item[(iv)] the sequence of output universes $H' \in \storelat(\mathbb{N}, \mathit{Universe})$.
\end{enumerate*}
The lifting to sequence of universes is inspired by \textsf{ReactiveML}~\cite{mandel-time-2015}.
The reaction rules are defined in Figure~\ref{react-rules}.
The rule \textsc{react} models the passing of time from one paused instant to the next.
Of interest, we notice that the values of the \texttt{single_space} variables are joined into all of the future universes.
We also observe that the two rules \textsc{react} and \textsc{exit-react} are exclusive on the termination condition.
We now discuss the side condition $causal(p)$ which performs the causality analysis of the program in each instant.

\subsection{Causality Analysis}
\label{causality-analysis}

Causality analysis is crucial to prove that spacetime programs are reactive, deterministic and extensive functions.
An example of non-reactive program is \texttt{when x |= y then f(write y) end}.
The problem is that if we add information in $y$, the condition \texttt{x |= y} might not be entailed anymore, which means that no derivation in the behavioral semantics is possible.
This is similar to emitting a signal in \textsf{Esterel} after we tested its absence.
Due to the lattice order on variables, we can however write on a value after an entailment condition, consider for example \texttt{when x |= y then f(write x) end}.
Whenever \texttt{x |= y} is entailed, it will stay entailed even if we later write additional information on $x$, so this program should be accepted.

The causality analysis symbolically executes an instant of a process, yielding the set of all symbolic paths reachable in an instant.
It also symbolically executes the paths of all branches generated in each instant.
For space reason, we only show the most important part of the causality analysis: the properties that a path must fulfil to be causal.
A path is a sequence of atomic statements $\langle a_1,\ldots, a_n \rangle$ where $a_i$ is defined as:
\begin{grammar}
<atom> ::= $x \vDash y$ | $f$($x_1^{w|r|\rw},\ldots,x_n^{w|r|\rw}$)
\end{grammar}
For example, the process \texttt{when $x$ |= $y$ then $f(x^r)$ else $g(x^r)$} generates two paths: $\langle x \vDash y, f(x^r) \rangle$ for the then-branch, and $\langle y \vDash x, g(x^r) \rangle$ for the else-branch.
A path $p$ is causal if for all atoms $a_i \in p$ the following two conditions hold.

First, for each entailment atom $a_i = x \vDash y$ we require:
\begin{equation}
\forall{z^b \in \mathit{Vars}(p_{i+1..|p|})},z = y \implies b = \texttt{r}
\label{read-after-entailment}
\end{equation}
with $\mathit{Vars}(p)$ the set of all variables in the path $p$.
It ensures all remaining accesses on $y$ to be read-only.

Second, for each function call $a_i = f(x_1^{b_1},..., x_n^{b_n})$ and each argument $x_k^{b_k}$ of $f$ we require:
\begin{equation}
\forall{z^b \in \mathit{Vars}(p_{i+1..|p|})}, x_k = z \land (b_k = \texttt{r} \lor b_k = \texttt{rw}) \implies b = \texttt{r}
\label{read-after-read}
\end{equation}
Whenever a variable is accessed with \texttt{read} or \texttt{readwrite}, it can only be read afterwards.
A consequence is that a variable cannot be accessed by two \texttt{readwrite} during a same instant.

\begin{definition}[Causal process]
A process is causal if for all its instants $i$, every path $p$ in the instant $i$ is causal ((\ref{read-after-entailment}) and (\ref{read-after-read}) hold).
\end{definition}

\subsection{Reactivity, Determinism and Extensiveness}
\label{proof-correctness}

We now only consider causal spacetime programs.
In this section, we sketch the proofs that the semantics of spacetime is deterministic, reactive and an extensive function during and across instants.
Importantly, these properties only hold if the underlying host functions meet the same properties.
The two first properties are typical of the synchronous paradigm and are defined as follows.

\begin{definition}[Determinism and reactivity]
For any state $\langle Q, p\rangle$, the derivation
\begin{displaymath}
0, \mathcal{L} \vdash \langle Q, p\rangle \xhookrightarrow[H]{H'} \langle Q', p' \rangle
\end{displaymath}
is deterministic (resp. reactive) if there is at most (resp. at least) one proof tree of the derivation.
\end{definition}

\begin{lemma}
The semantics of spacetime is reactive and deterministic.
\end{lemma}
\noindent
The proofs are given in Appendices~\ref{proof-reactivity} and~\ref{proof-determinism}.
They essentially verify the completeness and disjointness of the rules.

\begin{lemma}
The semantics of spacetime is extensive over its space during an instant.
\end{lemma}

\begin{proof}
Any write in the space is done through a variable declaration or a host function.
The declaration rules only add more information into the space by using the join operator $\sqcup$.
Otherwise, this property depends on the extensiveness of the host functions.
\end{proof}

To define the extensiveness property of a program across instants, we rely on the notion of \textit{observable space}.
Given a space $S \in Space$, its observable subset $obs(S) \subseteq S$ is the set of variables that can still be used in a future instant.

\begin{lemma}
Every observable variable is stored in either the queue or in a \texttt{single_space} variable.
\label{obs-spacetime}
\end{lemma}
\begin{proof}
The \texttt{world_line} variables are stored in a queue of nodes when pushed (rule~\textsc{react}).
In the case of a pruned node, the \texttt{world_line} variables are not observable since no child node can ever used their values again.
The \texttt{single_time} variables are reallocated in each instant, thus not observable in future instants.
\end{proof}

\begin{lemma}[Extensiveness]
Given a sequence of universes $H$ and two instant indices $i > j$, we have $H_i^S \vDash obs(H_j^S)$.
\end{lemma}

\begin{proof}
By Lemma~\ref{obs-spacetime}, it is sufficient to only look at the queue and \texttt{single_space} variables:
\begin{enumerate*}
\item[(i)] the queue is extensive by Definition~\ref{queueing-property} of the queueing strategy, and
\item[(ii)] the \texttt{single_space} variables are joined with their previous values (rule \textsc{react}), thus \texttt{single_space} variables that exist in $H_i^S$ and $H_j^S$ are ordered by induction on the instant indices.
\end{enumerate*}
Therefore our semantics is extensive with regard to the sequence of universes derived.
\end{proof}

\section{Constraint Programming in Spacetime}
\label{advanced-strategies}

In Section~\ref{spacetime-programming}, we defined a process generating an infinite binary search tree.
As the underlying structure of the state space is a lattice, the ``raw state space'' can be programmed by the user.
We demonstrate this fact by programming a process generating the state space of a constraint satisfaction problem (CSP).

A search strategy can be specialized or generic with regard to the state space.
For example, the strategy IDS (introduced in Section~\ref{peek-runtime}) can be reused on the CSP state space without modification.
As an additional example of generic search strategy, we consider limited discrepancy search (LDS) and its variants.
It can be combined effortlessly with IDS and the CSP state space generator.
We also introduce a branch and bound strategy which is bound to the CSP state space.
Overall, the goal is to show that search strategies can be developed independently from the state space while retaining their compositionality.


\subsection{Generating the CSP State Space}
\label{csp-state-space}

We consider a basic but practical solver using the propagate and search algorithm presented in Section~\ref{lattice-csp}.
\begin{lstlisting}[mathescape, language=bonsai]
class Solver {
  single_time ES consistent = unknown;
  ref world_line VStore domains;
  ref world_line CStore constraints;
  public Solver(VStore domains, CStore constraints) {$\ldots$}
  public proc search = par run propagation()$~\texttt{<>}~$run branch() end
  flow propagation =
    consistent$~\texttt{<-}~$constraints.propagate(readwrite domains);
    when consistent$\texttt{|=}$ true then prune end
  end
  flow branch =
    when unknown$~\texttt{|=}~$consistent then
      single_time IntVar x = failFirstVar(domains);
      single_time Integer v = middleValue(x);
      space constraints$~\texttt{<-}~$ x.le(v) end; // $x \leq v$
      space constraints$~\texttt{<-}~$ x.gt(v) end  // $x > v$
    end
  // Interface to the Choco solver.
  private IntVar failFirstVar(VStore domains) { $\ldots$ }
  private Integer middleValue(IntVar x) { $\ldots$ }  }
\end{lstlisting}
\noindent
This example introduces new elements of syntax:
\begin{enumerate*}
\item[(i)] the \texttt{ref} keyword which indicates that the variable name is an alias to a spacetime variable declared in another class,
\item[(ii)] the tell operator \texttt{x <- e} which is a syntactic sugar for \texttt{write x.join(e)}, the join operation $x = x \sqcup e$, and
\item[(iii)] the keywords \texttt{true}, \texttt{false} and \texttt{unknown} that are elements of the lattice $\mathit{ES}$ explained below.
\end{enumerate*}
We also remark that \texttt{read} annotations apply by default when not specified on variables.

The lattices \texttt{VStore} and \texttt{CStore} are respectively the variable store and the constraint store.
The constraint solver \textsf{Choco}~\cite{choco} is abstracted behind these two lattices and provides the main operations to propagate and branch on the state space.
The branching strategy is usually a combination of a function selecting a variable in the store (here \texttt{failFirstVar}) and selecting a value in the domain of the variable (here \texttt{middleValue}).
The two variables storing these results are annotated with \texttt{single_time} since they are only useful in the current instant.
We split the state space with the constraints $x \leq v$ and $x > v$.
In the implementation, this code is organized in a more modular way so we can assemble various parts of the branching strategies.

The lattice \texttt{ES} is defined as $\{\mathit{true}, \mathit{false}, \mathit{unknown}\}$ with the total order $\mathit{false} \vDash \mathit{true} \vDash \mathit{unknown}$.
It is used to detect if the current node of the CSP is a solution ($\mathit{true}$), a failed node ($\mathit{false}$) or if we do not know yet ($\mathit{unknown}$).
In the process \texttt{propagation}, we prune the current subtree if we reached a solution or failed node.



\subsection{Branch and Bound Search}
\label{bab}

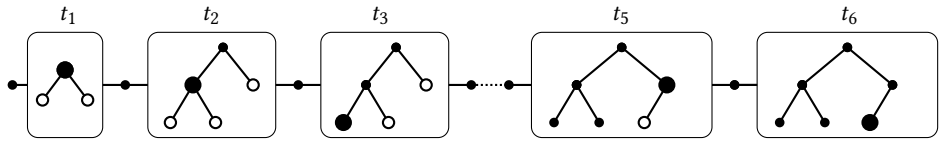
\begin{figure*}[t]
\begin{center}
\begin{tikzpicture}
\tikzstyle{time}=[draw,shape=circle,circle,fill,inner sep=1pt]
\tikzstyle{current}=[draw,shape=circle,circle,fill,inner sep=2pt]
\tikzstyle{explore}=[draw,shape=circle,circle,fill=white,inner sep=1.5pt]

\begin{scope}[yshift=-1.8cm]
\node[anchor=south] at (3.75,0) {$t_1$};
\node[anchor=south] at (5.65,0) {$t_2$};
\node[anchor=south] at (7.9,0) {$t_3$};
\node[anchor=south] at (11.1, 0) {$t_5$};
\node[anchor=south] at (14.1325, 0) {$t_6$};
\end{scope}

\begin{scope}[xshift=0cm,yshift=-6cm]
  \draw[rounded corners] (3.2,4.2) rectangle (4.2,2.8);

  \draw[thick] (4.2,3.5) node{} -- (4.5,3.5) node[time]{} -- (4.8,3.5);
  \draw[thick] (3,3.5) node[time]{} -- (3.2,3.5);

  \begin{scope}[xshift=0.7cm]
  \draw[thick] (3,3.7) node[current]{} -- (2.7,3.3) node[explore]{};
  \draw[thick] (3,3.7) node[current]{} -- (3.3,3.3) node[explore]{};
  \end{scope}
\begin{scope}[xshift=2.3cm]
  \draw[rounded corners] (2.5,4.2) rectangle (4.2,2.8);

  \draw[thick] (4.2,3.5) node{} -- (4.5,3.5) node[time]{} -- (4.8,3.5);

  \begin{scope}[xshift=0.5cm]

  \draw[thick] (3,4) node[time]{} -- (2.6,3.5) node[current]{};
  \draw[thick] (3,4) node[time]{} -- (3.4,3.5) node[explore]{};

  \draw[thick] (2.6,3.5) node[time]{} -- (2.3,3) node[explore]{};
  \draw[thick] (2.6,3.5) node[time]{} -- (2.9,3) node[explore]{};
  \end{scope}

\begin{scope}[xshift=2.3cm]
  \draw[rounded corners] (2.5,4.2) rectangle (4.2,2.8);
  \draw[thick] (4.2,3.5) -- (4.5,3.5) node[time]{};
  \draw[thick,densely dotted] (4.5,3.5) -- (5,3.5);
  \draw[thick] (5,3.5) node[time]{} -- (5.3,3.5);

  \begin{scope}[xshift=0.5cm]

  \draw[thick] (3,4) node[time]{} -- (2.6,3.5) node[time]{};
  \draw[thick] (3,4) node[time]{} -- (3.4,3.5) node[explore]{};

  \draw[thick] (2.6,3.5) node[time]{} -- (2.3,3) node[current]{};
  \draw[thick] (2.6,3.5) node[time]{} -- (2.9,3) node[explore]{};
  \end{scope}

\begin{scope}[xshift=3.5cm]
  \draw[rounded corners] (1.8,4.2) rectangle (4.2,2.8);
  \draw[thick] (4.2,3.5) node{} -- (4.5,3.5) node[time]{} -- (4.8,3.5);

  \draw[thick] (3,4) node[time]{} -- (2.4,3.5) node[time]{};
  \draw[thick] (3,4) node[time]{} -- (3.6,3.5) node[current]{};

  \draw[thick] (2.4,3.5) node[time]{} -- (2.7,3) node[time]{};
  \draw[thick] (2.4,3.5) node[time]{} -- (2.1,3) node[time]{};

  \draw[thick] (3.6,3.5) node[time]{} -- (3.3,3) node[explore]{};

\begin{scope}[xshift=3cm]
  \draw[thick] (1.5,3.5) node[time]{} -- (1.8,3.5);
  \draw[rounded corners] (1.8,4.2) rectangle (4.2,2.8);

  \draw[thick] (3,4) node[time]{} -- (2.4,3.5) node[time]{};
  \draw[thick] (3,4) node[time]{} -- (3.6,3.5) node[time]{};

  \draw[thick] (2.4,3.5) node[time]{} -- (2.7,3) node[time]{};
  \draw[thick] (2.4,3.5) node[time]{} -- (2.1,3) node[time]{};

  \draw[thick] (3.6,3.5) node[time]{} -- (3.3,3) node[current]{};
\end{scope}
\end{scope}
\end{scope}
\end{scope}
\end{scope}

\end{tikzpicture}
\end{center}
\caption{Combination of bounded depth and bounded discrepancy search.}
\label{bds-lds-fig}
\vspace{-0.2cm}
\end{figure*}

Branch and bound (BAB) is an algorithm to find the optimal solution of a CSP according to an objective function.
BAB reasons over the whole search tree by keeping track of the \textit{best solution} obtained so far, in contrast to propagation which operates on a single node at a time.
It is implemented in the following class \texttt{MinimizeBAB}.

\begin{lstlisting}[mathescape,language=bonsai]
public class MinimizeBAB {
  ref world_line VStore domains;
  ref single_time ES consistent;
  ref single_space IntVar x;
  single_space LMin obj = bot;
  public MinimizeBAB(VStore domains, ES consistent, IntVar x) {$...$}
  public proc solve = par run minimize() <> run yield_objective() end
  proc minimize =
    loop
      when consistent== true then
        single_space LMin pre_obj = new LMin(x.getLB());
        pause;
        obj <- pre_obj;
      else pause end
    end
  flow yield_objective =
    consistent <- updateBound(write domains, write x, read obj)
  static ES updateBound(VarStore domains, IntVar x, LMin obj) { $...$ }
\end{lstlisting}
\noindent
Along with the current variable store, we have the variable \texttt{x} to be minimized and its current best bound \texttt{obj} of type \texttt{LMin}.
The \texttt{single_space} attribute indicates that the bound \texttt{obj} is global to the search tree and will not be backtracked.
The class has two main processes:
\begin{enumerate*}
\item[(i)] \texttt{minimize} strengthens the bound \texttt{obj} with the value obtained in the previous solution node, and
\item[(ii)] \texttt{yield_objective}, through the function \texttt{updateBound}, interfaces with \textsf{Choco} to update \texttt{domains} with $x < obj$, so the next solution we find have a better bound.
\end{enumerate*}

There is an important detail to notice: we use a temporary variable \texttt{pre_obj} to store the latest bound instead of updating \texttt{obj} directly.
Interestingly, if we do not, the causality analysis will fail since we have a cyclic dependency in the data: \texttt{obj} depends on \texttt{domains} and vice versa.
Fortunately, the causality analysis prevents us from having a bug: adding the current bound in the CSP would turn a solution node into a failed node.


\subsection{Limited Discrepancy Search and Variants}

For some problems, the branching strategy can order the branches such that the left one is more likely to reach a solution first.
Limited discrepancy search (LDS) was introduced to take advantage of this ordering property.
It is based on the notion of \textit{discrepancies} which is the number of right branches taken to reach a leaf node.
In its original version~\cite{harvey-limited-1995}, LDS successively increases the number of discrepancies a branching strategy can take by restarting the exploration of the full tree.
The paths with $0$, $1$, $2$ and $3$ discrepancies in a tree of depth $3$ are given as follows:
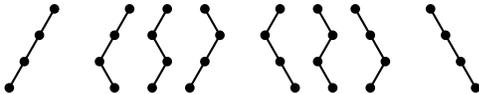
\begin{center}
\begin{tikzpicture}
\tikzstyle{time}=[draw,shape=circle,circle,fill,inner sep=1pt]
\tikzstyle{current}=[draw,shape=circle,circle,fill,inner sep=2pt]
\tikzstyle{explore}=[draw,shape=circle,circle,fill=white,inner sep=1.5pt]

\draw[thick] (3,3) node[time]{} -- (2.8,2.65) node[time]{} -- (2.6, 2.3) node[time]{} -- (2.4, 1.95) node[time]{};

\begin{scope}[xshift=1cm] \draw[thick] (3,3) node[time]{} -- (2.8,2.65) node[time]{} -- (2.6, 2.3) node[time]{} -- (2.8, 1.95) node[time]{}; \end{scope}
\begin{scope}[xshift=1.5cm] \draw[thick] (3,3) node[time]{} -- (2.8,2.65) node[time]{} -- (3, 2.3) node[time]{} -- (2.8, 1.95) node[time]{}; \end{scope}
\begin{scope}[xshift=2cm] \draw[thick] (3,3) node[time]{} -- (3.2,2.65) node[time]{} -- (3, 2.3) node[time]{} -- (2.8, 1.95) node[time]{}; \end{scope}

\begin{scope}[xshift=3cm] \draw[thick] (3,3) node[time]{} -- (2.8,2.65) node[time]{} -- (3, 2.3) node[time]{} -- (3.2, 1.95) node[time]{}; \end{scope}
\begin{scope}[xshift=3.5cm] \draw[thick] (3,3) node[time]{} -- (3.2,2.65) node[time]{} -- (3, 2.3) node[time]{} -- (3.2, 1.95) node[time]{}; \end{scope}
\begin{scope}[xshift=4cm] \draw[thick] (3,3) node[time]{} -- (3.2,2.65) node[time]{} -- (3.4, 2.3) node[time]{} -- (3.2, 1.95) node[time]{}; \end{scope}

\begin{scope}[xshift=5cm] \draw[thick] (3,3) node[time]{} -- (3.2,2.65) node[time]{} -- (3.4, 2.3) node[time]{} -- (3.6, 1.95) node[time]{}; \end{scope}
\end{tikzpicture}
\end{center}
The first iteration generates the leftmost path, the second iteration allows one discrepancy, and so on.
The search is complete if the discrepancy limit is not reached during one iteration.
An iteration of LDS is programmed in spacetime as follows:
\begin{lstlisting}[mathescape,language=bonsai]
public class BoundedDiscrepancy {
  single_space LMax limit;
  world_line LMax dis = new LMax(0);
  public BoundedDiscrepancy(LMax limit) { $...$ }
  public flow bound =
    space nothing end;
    when dis |= limit then prune
    else space readwrite dis.inc() end end
  end }
\end{lstlisting}
\noindent
Initially, the discrepancy counter \texttt{dis} is set to $0$.
The left branch is always taken, which we represent with a neutral \texttt{space nothing end} statement.
The right branch is taken only if the discrepancies counter is less than the limit, otherwise we prune this branch.
We can restart this search with the same technique as the one used for IDS (Section~\ref{peek-runtime}).

A drawback of LDS is that at each iteration $k$, it re-explores all paths with $k$ \textit{or less} discrepancies.
In~\cite{Korf96improvedlimited}, Korf proposes an improved version of LDS (ILDS) where only paths with exactly $k$ discrepancies are explored.
We provide a library of reusable improved LDS strategies including ILDS, depth-bounded discrepancy search (DDS)~\cite{walsh-depth-bounded-1997} and LDS variants~\cite{prosser-limited-2011} in the implementation.

In addition to creating a search strategy from scratch, we often need to assemble existing strategies to obtain the best of two approaches.
For example, the combination of LDS with IDS is discussed in~\cite{harvey-limited-1995}, as well as the combination of DDS with IDS in~\cite{walsh-depth-bounded-1997}.
These combinations can be easily programmed in spacetime; we obtain the first by combining \texttt{BoundedDepth} and \texttt{BoundedDiscrepancy}:
\begin{lstlisting}[mathescape, language=bonsai]
module BoundedDepth bd = new BoundedDepth(new LMax(2));
module BoundedDiscrepancy bdis =
  new BoundedDiscrepancy(new LMax(1));
par run db.bound() $\texttt{<>}$ run bdis.bound() end
\end{lstlisting}
\noindent
The result of this combination is shown in Figure~\ref{bds-lds-fig}.
Similarly we can use the disjunctive parallel operator \texttt{||} to obtain their union.
What's more, we can apply this strategy to the CSP state space, possibly augmented with the BAB process, in the very same way.

\section{Implementation}
\label{implementation}

The compiler of spacetime performs static analyses to ensure well-formedness of the program.
It includes common analyses and transformations on synchronous programs such as causality analysis, detection of instantaneous loop and reincarnation~\cite{tardieu-loops-2005,esterel-compilation}.
Specifically in spacetime, we ensure that every statement \texttt{space $p$} has an instantaneous body and does not contain nested \texttt{space} or \texttt{prune} statements.
In addition, we provide several analyses to integrate Java and spacetime in a coherent way, especially for initializing objects with existing spacetime variables (keyword \texttt{ref}).
These analyses are out of scope in this paper, but we provide a comprehensive list of the analyses in the file \texttt{src/errors.rs} of the implementation.

As shown in Section~\ref{peek-runtime}, every spacetime statement is mapped to a \textit{synchronous combinator} encoding its behavior at runtime.
Synchronous combinators are also used in the context of synchronous reactive programming---basically Esterel without reaction to absence---in the Java library \textsf{SugarCubes}~\cite{boussinot-java-2000,susini-reactive-2006}.
In this section, we overview how these combinators are scheduled in the runtime.

\textbf{Replicating.} Every spacetime program presented in this paper, as well as the experiments below, are available in the repository \url{https://github.com/ptal/bonsai/tree/PPDP19}.

\subsection{Scheduling Algorithm}

The main purposes of the runtime are to dynamically schedule concurrent processes, to retain the state of the program from an instant to the next, and to push and pop variables onto the queue.
To achieve these goals, we extend the structures introduced in the behavioral semantics (Section~\ref{space-structure}) to incorporate \textit{access counters} and a \textit{suspended completion code}.

Firstly, we equip every variable with an \textit{access counter} $(w, \rw, r) \in \mathit{LMin}^3$ where $w$ is the numbers of \texttt{write}, $\rw$ of \texttt{readwrite} and $r$ of \texttt{read} accesses that can still happen on a variable in the current instant.
As suggested by the lattice $\mathit{LMin}$, these counters are decreased whenever the corresponding access is performed.
We extend the poset $\mathit{Var}$ to access counters: $\{\top\} \cup (\mathit{spacetime} \times \mathit{Value} \times \mathit{LMin}^3)$.

Secondly, given a variable $x$ and its access counter $(w, \rw, r)$, we say that a process is \textit{suspended} if it needs to perform a \texttt{readwrite} on $x$ when $w > 0$, or to read $x$ when $w > 0$ or $\rw > 0$.
A process cannot be suspended on a write access.
Whenever a process is stuck, the flow of control is given to another process.
We add this additional stuck status in the set of completion codes $\mathit{Compl}$ with the code $3$.


In order to schedule processes, the runtime performs a \textit{can} and \textit{cannot} analyses over the program.
The \textit{can analysis} computes an upper bound on the counters: the numbers of accesses that can still happen on each variable in the current instant.
The \textit{cannot analysis} decreases counters by invalidating parts of the program that cannot be executed.

Consider the following spacetime program ($x,y \in LMax$):
\begin{lstlisting}[language=bonsai]
when x |= y then f(write x, read y) else g(read x, write y) end
\end{lstlisting}
\noindent
Initially, the counters of $x$ and $y$ are both set to $(1,0,1)$.
Therefore, we cannot decide the entailment of $x \vDash y$ because its result might change due to future writes on $x$ or $y$.
However, we observe that if $x \vDash y$ holds then we can only write on $x$, which cannot change the entailment result.
Similarly if $x \nvDash y$ holds, we can only write on $y$.
To unlock such a situation, the \textit{cannot analysis} decreases the counters of the variables with unreachable read/write accesses.
Thanks to the causality analysis, a deadlock situation cannot happen since every access in every path is well-ordered.



The algorithm scheduling an instant alternates between the execution of the process, and the decrement of access counters with the $cannot$ analysis.\footnote{A sketch of this algorithm is available in Appendix~\ref{scheduling-algorithm}.}
The mechanics of this scheduling algorithm is close to the one of \textsf{SugarCubes}~\cite{boussinot-java-2000} and \textsf{ReactiveML}~\cite{mandel-reactiveml:-2005}.

\subsection{Experiments}
\label{experiments}
\begin{table}
\begin{tabular}{c|c c c}
Problem & Spacetime & \textsf{Choco} & Factor \\
\hline
13-Queens & 16.4s (62946n/s) & 5.3s (194304n/s) & 3.1 \\
14-Queens & 89.9s (62020n/s) & 30.6s (182218n/s) & 2.9 \\
15-Queens & 528.2s (60972n/s) & 185.2s (173816n/s) & 2.85 \\
Golomb Ruler 10 & 1.8s (17407n/s) & 1s (31154n/s) & 1.8 \\
Golomb Ruler 11 & 40.1s (14186n/s) & 27.2s (20888n/s) & 1.47 \\
Golomb Ruler 12 & 425.8s (10871n/s) & 279.8s (16541n/s) & 1.52 \\
Latin Square 60 & 19s (155n/s) & 17.1s (172n/s) & 1.10 \\
Latin Square 75 & 61.2s (73n/s) & 57.9s (77n/s) & 1.06 \\
Latin Square 90 & 150.3s (44n/s) & 147.8s (45n/s) & 1.02 \\
\end{tabular}
\caption{Comparison of spacetime and \textsf{Choco} on the resolution time and nodes-per-second (n/s).}
\label{experiments-table}
\vspace{-0.6cm}
\end{table}

We terminate this section with a short experimental evaluation.
The experiments were run on a 1.8GHz Intel(R) Core(TM) i7-8550U processor running GNU/Linux.
A warm-up time of about $30s$ was performed on the JVM before any measure was recorded.

We select three CSPs to test the overhead of a spacetime strategy in comparison to the same hard-coded \textsf{Choco} strategy.
The propagation engine is the one of \textsf{Choco} in both cases.
As shown in Table~\ref{experiments-table}, the overhead factor of spacetime varies from almost $1$ to at most $3.1$ depending on the problem to solve.
To obtain a \textit{search intensive} algorithm, we search for all solutions of the N-Queens problem which has only three constraints to propagate in each node.
This is the worst-case scenario for spacetime since the number of nodes is directly linked to the number of reactions of a spacetime program, and thus its overhead factor.
We also consider a \textit{propagation intensive} algorithm by searching for a single solution of a Latin Square problem which has a large number of constraints.
To find a solution, the search never backtracks so the number of nodes is few.
This explains the small overhead factor of spacetime which is almost $1$.
Finally, we evaluate a branch and bound (BAB) search strategy on the Golomb Ruler problem.
BAB finds the best solution of an optimization problem, and thus explores a large tree.
In this case, the overhead factor of spacetime drops to $1.5$ thanks to the more realistic balance between search and propagation.

As for the correctness, spacetime always finds the same number of nodes, solutions and failures than \textsf{Choco}, as well as the same lower bounds for optimization problems (Golomb ruler).
It indicates that the exact same search tree is explored.


\section{Related Work}
\label{related-work}

We review two families of search languages: constraint logic programming and combinator-based search languages.
Afterwards, we discuss the independent issue of integrating arbitrary data into imperative synchronous languages.

\subsection{Constraint Logic Programming}

Constraint logic programming (CLP)~\cite{jaffar-constraint-1987} is a paradigm extending logic programming with constraints.
We can program search strategies by using the backtracking capabilities of logic programming.
CLP systems such as \textsf{GNU-Prolog}~\cite{codognet-compiling-1996,diaz-implementation-2012} and \textsf{Eclipse}~\cite{apt-constraint-2007,schimpf-eclipse-2012} propose various built-in blocks to construct a customized search strategy.
Although CLP is an elegant formalism, it suffers from three drawbacks:
\begin{enumerate}
\item There is no mechanism to compose search strategies.
\item Global state, such as a node counter, is programmed via system dependent non-backtrackable mutable state libraries.
\item It is bound to the evaluation strategy of Prolog, which for example means that LDS with highest-occurrence discrepancies cannot be easily implemented.\footnote{See the documentation of \textsf{Eclipse} at \url{http://eclipseclp.org/doc/bips/lib/fd_search/search-6.html}.}
\end{enumerate}

The \texttt{tor/2} predicate~\cite{Schrijvers:2014:TMS:2608851.2608962} tackles the compositionality issue of CLP systems.
It proposes to replace the disjunctive Prolog predicate \texttt{;/2} by a \texttt{tor/2} predicate which, in addition to creating two branches in the search tree, is a synchronization point.
Two search strategies defined with \texttt{tor/2} can be merged with the predicate \texttt{tor_merge/2}.
This extension allows the user to program various strategies independently and to assemble them.
However, the search predicates are not executed concurrently, thus two search strategies cannot be interleaved and communicate over a shared variable.
For example, the processes \texttt{Solver.search} and \texttt{MinimizeBAB.solve} must be interleaved because they communicate over the variables \texttt{consistent} and \texttt{domains}.

\subsection{Search Combinators}


Early constraint search languages appeared around 1998 with \textsf{Localizer}~\cite{michel-localizer:-1999}, \textsf{Salsa}~\cite{salsa-2002} and \textsf{OPL}~\cite{van-hentenryck-opl-1999,van-hentenryck-search-2000}.
More recent approaches include \textsf{Comet}~\cite{VanHentenryck:2005:CLS:1121598} (successor of \textsf{Localizer}), the search combinators~\cite{search2013} and its subset \textsf{MiniSearch}~\cite{rendl-minisearch:2015}.
\textsf{Comet} and \textsf{Localizer} are specialized to local search, a non-exhaustive form of constraint solving.
Local search languages differ because their programs are not necessarily extensive and are not always based on backtracking search.
Search combinators mostly focus on the control part of search and it is interesting to take an example (from~\cite{search2013}):
\begin{displaymath}
\begin{array}{l}
id(s) \eqdef ir(depth, 0, +, 1, \infty, s)\\
ir(p, l, \oplus, i, u, s) \eqdef let(n, l, restart(n \leq u, \\
  \qquad\quad and([assign (n, n \oplus i), limit (p \leq n, s)])))
\end{array}
\end{displaymath}
\noindent
The combinator $id$ is an iterative depth-first search (IDS)~\cite{Korf85depth-firstiterative-deepening} that restarts a strategy $s$ by increasing the depth limit.
The pattern of iteratively restarting the search is encapsulated in a combinator $ir$ where the strategy $s$ is restarted until we reach a limit $n \leq u$.
To summarize, $n$ is an internal counter initialized at $l$, and increased by $n \oplus i$ on each restart.
They show that LDS is just another case of the combinator $ir$ with discrepancies.

In search combinators, the search strategy is written \textit{vertically}: each strategy is encapsulated in another strategy.
In spacetime, we compose search strategy \textit{horizontally}: each strategy is executed concurrently (``next to'') another strategy.
We believe that both vertical and horizontal compositionality is required in order to achieve high re-usability of search strategies.

A drawback of combinators-based languages is that they rely on data from the constraint solver, and the interactions with the host language are not formalized.
In particular, it is not possible that two search strategies safely communicate over shared variables.

\subsection{Arbitrary Data in Synchronous Languages}

Signals in \textsf{Esterel} are Boolean values, which are limited when processes need to communicate more complex information.
This is why they bring the notions of valued signals and variables for storing non-Boolean values~\cite{berry-esterel-2000,esterel-technologies-esterel-2005}.
However, they are more restricted than pure signal: testing the value of a signal is only possible when all emissions have been performed, and variables must not be shared for writing across processes.
Sequentially constructive \textsf{Esterel} (SCEst)~\cite{smyth-scest:-2018} brings variables to \textsf{Esterel} that can be used across processes.
The main idea is that any value must be manipulated following an \textit{init-update-read} cycle within an instant.
This is similar to our way to schedule \textit{write-readwrite-read}, but there is no notion of order between values in SCEst.
Therefore we can use destructive assignment similarly to sequential languages.
In spacetime, the choice of lattices as the underlying data model comes from CCP and is more suited for constraint programming.
In this respect, lattice-based variables unify the notions of signals, valued signals and variables of \textsf{Esterel}.

\textsf{ReactiveML} merges the imperative synchronous and functional paradigms without negative ask~\cite{mandel-reactiveml:-2005}.
An advantage is that we can manipulate arbitrary functional data.
Note that the addition of mutable states to \textsf{ReactiveML} is not deterministic~\cite{mandel-conception-2006}.

Default TCC~\cite{saraswat-timed-1999} is TCC with negative ask.
It views an instant as a set of closure operators, one for each assumption on the result of the ask statements.
A weakness of default TCC is to speculate on the result of the negative asks, which is implemented by backtracking inside an instant if its guess was wrong~\cite{saraswat-jcc:-2003}.
This is also problematic for external functions that produce side-effects.

\section{Conclusion}

Concurrent constraint programming (CCP) is a theoretical paradigm that formalizes concurrent logic programming inspired by constraint logic programming~\cite{ueda-logic/constraint-2017}.
Unfortunately, this marriage is incomplete since backtracking, available in constraint logic programming, is not incorporated in CCP.
We believe that the missing piece is the notion of logical time, as it appears in the synchronous paradigm, and it fostered the development of spacetime.

In the first part of this paper, we argued that logical time is a suitable device to conciliate concurrency and backtracking.
The main underlying idea is captured as follows: a search strategy explores one node of the search tree per logical instant.
In particular, we took the example of constraint solving in which designing search strategies is crucial to solve a CSP efficiently.
We developed several search strategies in a modular way, and showed that they can be composed to obtain a new one.
As a result, spacetime improves on the compositionality issues faced by developers of search strategies.

In the second part of this paper, we developed the foundations of spacetime by extending the behavioral semantics of \textsf{Esterel} to lattice-based variables and backtracking.
We proved that the semantics is deterministic, reactive and that a spacetime program only accumulates more and more information during and across instants (extensiveness).

Further developments of spacetime include static compilation such as in Esterel~\cite{esterel-compilation} to improve efficiency, development in a proof assistant of the reactivity, determinism and extensiveness proofs, and formalization of a precise connection between the operational semantics (runtime) and the behavioral semantics.
Furthermore, a natural extension of spacetime is to reify the queue inside the language itself instead of relying on the host language.
The key idea of this extension is to merge the time hierarchy of synchronous languages~\cite{gemunde-clock-2013,mandel-time-2015} and the space hierarchy induced by deep guards in logic programming~\cite{jaha91} and Oz computation spaces~\cite{oz-space-schulte}.
First-class queue will allow users to program restart-based search strategies directly in spacetime instead of partly relying on the host language.
Preliminary extension of the compiler indicates that this approach is feasible.
Finally, although we applied spacetime to constraint programming, the notion of constraints is not built-in since we rely on lattice abstractions.
Therefore, we firmly believe that spacetime is suitable to express strategies in other fields tackling combinatorial exploration such as in satisfiability modulo theories (SMT), model checking and rewriting systems.

\begin{acks}
I would like to thank the anonymous reviewers for their useful comments.
I am indebted to Carlos Agon who gave me the freedom to pursue this research and for fruitful discussions about the spacetime paradigm.
My special thanks go to Cl\'{e}ment Poncelet, David Cachera, Charlotte Truchet and Yinghan Ling for their helpful comments and suggestions on this paper.
This work was partially supported by \grantsponsor{}{Agence Nationale de la Recherche}{} under Grant No.~\grantnum{}{15-CE25-0002} (Coverif).
A substantial part of this work was done at the Institute for Research and Coordination in Acoustics/Music (IRCAM) and Sorbonne University, Paris.
\end{acks}



\appendix
\section{Appendix}



\subsection{Proof of Reactivity}
\label{proof-reactivity}

\begin{lemma}
The semantics of spacetime is reactive.
\end{lemma}
\begin{proof}
Given a program $p$, we can always choose a rule to apply, this is checked by verifying the completeness of the side conditions on rules applying to a same program.
\begin{itemize}
  \item Axioms \textsc{nothing}, \textsc{pause}, \textsc{stop} and \textsc{prune} are always reactive because they do not have side conditions.
  \item Axioms \textsc{space} and \textsc{space-pruned} derives both pruned and non-pruned branch.
  In \textsc{space}, enforcing instantaneousness of the body and forbidding writes in \texttt{single_space} or \texttt{single_time} variables can be statically checked at compile-time.
  \item \texttt{loop} is reactive if the loop is not instantaneous, this can be statically checked at compile-time.
  \item \textsc{when-true} and \textsc{when-false} are reactive since the entailment operation only maps to $true$, $false$ or $unknown$ (last both are handled in \textsc{when-false}).
  Moreover, due to the causality analysis (property~\ref{read-after-entailment}), the entailment result cannot further change during the derivation.
  \item Given \texttt{$p$;$q$}, \textsc{enter-seq} and \textsc{next-seq} are complete on the completion code of $p$: $U'^k = 0 \lor \lnot (U'^k = 0)$ is a tautology.
  \item \textsc{par$^\lor$} is always applicable.
  \item Given \texttt{$p$ || $q$}, \textsc{par$^\land$} and \textsc{exit-par$^\land$} are complete on the completion code of $p$ and $q$.
  We have $(U'^k \neq 0 \land U''^k \neq 0) \lor (U'^k = 0 \lor U''^k = 0)$ a tautology.
  \item \textsc{var-decl$\circlearrowleft$} is always applicable.
  \item \textsc{start-var-decl$\rightarrow\downarrow$} and \textsc{resume-var-decl$\rightarrow\downarrow$} are complete (either we have a location or a variable name).
  For \textsc{resume-var-decl$\rightarrow\downarrow$}, the function $pop$ returns $\bot$ if the queue is empty, any variable not defined in a space is mapped to $\bot$ as well (cf. Section~\ref{space-structure}), so the initialization of a \texttt{world_line} variable is reactive.
  \item \textsc{hcall} depends on the semantics of the host language.
  The causality analysis guarantees that the function is only called if all its variables can be safely accessed:
  \begin{itemize}
    \item A \texttt{write} access is always possible.
    \item For \texttt{read} access, we ensure this variable cannot be written anymore in the future (by property~\ref{read-after-read}).
    \item For \texttt{readwrite} access, only one of such access can happen in an instant (by property~\ref{read-after-read}), and it must happen after every write on this variable.
  \end{itemize}
  \item \textsc{react} and \textsc{exit-react} are complete on the termination condition.
  We have $(U'^k = 1 \text{ and $Q'$ is not empty}) \lor (U'^k \neq 1 \text{ or $Q'$ is empty})$ a tautology.
\end{itemize}
\end{proof}

\renewcommand{\algorithmicrequire}{\textbf{Input:}}
\renewcommand{\algorithmicensure}{\textbf{Output:}}
\begin{algorithm*}
\caption{Runtime engine}
\label{runtime-engine}
\begin{algorithmic}[1]
\Require A spacetime program $p$, a space $S \in Space$ and a queue $Q$.
\Ensure The triple $\langle p, S, Q \rangle$ such that either $p$ is stopped or terminated, or $Q$ is empty.
\Procedure{execute}{$p, S, Q$}
\State $k \leftarrow 1$ \Comment{Completion code initialized to pause.}
\If{First instant}
  \State $Q \leftarrow push(Q, \{\bot\})$ \Comment{Bootstrap the queue with a single element.}
\EndIf
  \While{$k = 1 \land Q$ is not empty}
    \State $\langle Q, S^{\downarrow}\rangle \leftarrow pop(Q)$
    \State $S \leftarrow can(p,S)$ \Comment{We compute an upper bound on the access counters.}
    \State $\langle p, S, B, k\rangle \leftarrow executeInstant(p, S)$
    \State $Q \leftarrow push(Q, B)$
  \EndWhile
\State \textbf{return} $(p, S, Q)$
\EndProcedure
\end{algorithmic}
\end{algorithm*}

\renewcommand{\algorithmicrequire}{\textbf{Input:}}
\renewcommand{\algorithmicensure}{\textbf{Output:}}
\begin{algorithm*}
\caption{Runtime execution of one instant}
\label{exec-instant}
\begin{algorithmic}[1]
\Require A spacetime program $p$ and a space $S \in Space$.
\Ensure The tuple $\langle p, S, B, k \rangle$ such that $B$ is the set of branches and $k$ the completion code.
\Procedure{executeInstant}{$p, S$}
\State $k \leftarrow 3$ \Comment{Completion code initialized to stuck.}
  \While{$k = 3$}
    \State $\langle p, S, B, k\rangle \leftarrow executeProcess(p, S)$
    \If{$k = 3$}
      \State $\langle p, S \rangle \leftarrow cannot(p, S)$ \Comment{We decrease the upper bound on the access counters}
    \EndIf
  \EndWhile
\State \textbf{return} $(p, S, B, k)$
\EndProcedure
\end{algorithmic}
\end{algorithm*}

\subsection{Proof of Determinism}
\label{proof-determinism}

\begin{lemma}
The semantics of spacetime is deterministic.
\end{lemma}

\begin{proof}
We check that for every rule, at most one rule can be applied to any process $p$, this is checked by verifying that rules on a same statement are exclusive to each other.
\begin{itemize}
  \item Rules \textsc{nothing}, \textsc{pause}, \textsc{stop}, \textsc{prune}, \textsc{loop}, \textsc{var-decl$\circlearrowleft$} and \textsc{par$^\lor$} are deterministic because only one rule can apply.
  \item Axioms \textsc{space} and \textsc{space-pruned} are exclusive on the kind of branch, so it is deterministic.
  \item \textsc{when-true} and \textsc{when-false} are deterministic since the side conditions on the entailment are exclusive.
  \item Given \texttt{$p$;$q$}, \textsc{enter-seq} and \textsc{next-seq} are exclusive on the completion code of $p$.
  \item Given \texttt{$p$ || $q$}, \textsc{par$^\land$} and \textsc{exit-par$^\land$} are exclusive on the completion code of $p$ and $q$.
  \item Due to the disjointness of the sets $\nameset$ and $Loc$, we can only apply either \textsc{start-var-decl$\rightarrow\downarrow$} or \textsc{resume-var-decl$\rightarrow\downarrow$}.
  \item \textsc{hcall} is deterministic if the semantics of the host language is deterministic.
  \item \textsc{react} and \textsc{exit-react} are exclusive on the termination condition.
\end{itemize}
\end{proof}

\subsection{Scheduling Algorithm}
\label{scheduling-algorithm}

We divide the runtime algorithm into two parts: the execution of several instants in Algorithm~\ref{runtime-engine} and the execution of an instant in Algorithm~\ref{exec-instant}.

The first algorithm implements the rules \textsc{react} and \textsc{exit-react} of the behavioral semantics.
In addition, it initializes the access counters before each instant with the $can$ function.

The second algorithm is the scheduler of the processes inside an instant.
It alternates between $executeProcess$ and $cannot$ until the process is not suspended anymore.
Consequently, this function never returns a suspended completion code.
The function $executeProcess$ is implemented following the same mechanics than \textsf{SugarCubes}~\cite{boussinot-java-2000} and some ideas from \textsf{ReactiveML}~\cite{mandel-reactiveml:-2005}.


\begin{thebibliography}{62}


\ifx \showCODEN    \undefined \def \showCODEN     #1{\unskip}     \fi
\ifx \showDOI      \undefined \def \showDOI       #1{#1}\fi
\ifx \showISBNx    \undefined \def \showISBNx     #1{\unskip}     \fi
\ifx \showISBNxiii \undefined \def \showISBNxiii  #1{\unskip}     \fi
\ifx \showISSN     \undefined \def \showISSN      #1{\unskip}     \fi
\ifx \showLCCN     \undefined \def \showLCCN      #1{\unskip}     \fi
\ifx \shownote     \undefined \def \shownote      #1{#1}          \fi
\ifx \showarticletitle \undefined \def \showarticletitle #1{#1}   \fi
\ifx \showURL      \undefined \def \showURL       {\relax}        \fi
\providecommand\bibfield[2]{#2}
\providecommand\bibinfo[2]{#2}
\providecommand\natexlab[1]{#1}
\providecommand\showeprint[2][]{arXiv:#2}

\bibitem[\protect\citeauthoryear{Apt}{Apt}{1999}]%
        {apt-essence-1999}
\bibfield{author}{\bibinfo{person}{Krzysztof~R. Apt}.}
  \bibinfo{year}{1999}\natexlab{}.
\newblock \showarticletitle{The essence of constraint propagation}.
\newblock \bibinfo{journal}{\emph{Theoretical computer science}}
  \bibinfo{volume}{221}, \bibinfo{number}{1-2} (\bibinfo{year}{1999}),
  \bibinfo{pages}{179--210}.
\newblock
\urldef\tempurl%
\url{https://doi.org/10.1016/S0304-3975(99)00032-8}
\showDOI{\tempurl}

\bibitem[\protect\citeauthoryear{Apt and Wallace}{Apt and Wallace}{2007}]%
        {apt-constraint-2007}
\bibfield{author}{\bibinfo{person}{Krzysztof~R. Apt} {and} \bibinfo{person}{M
  Wallace}.} \bibinfo{year}{2007}\natexlab{}.
\newblock \bibinfo{booktitle}{\emph{Constraint logic programming using
  {ECLiPSe}}}.
\newblock \bibinfo{publisher}{Cambridge University Press}.
\newblock
\showISBNx{9780521866286}


\bibitem[\protect\citeauthoryear{Beck}{Beck}{2007}]%
        {beck-solution-guided-2007}
\bibfield{author}{\bibinfo{person}{J.~Christopher Beck}.}
  \bibinfo{year}{2007}\natexlab{}.
\newblock \showarticletitle{Solution-guided multi-point constructive search for
  job shop scheduling}.
\newblock \bibinfo{journal}{\emph{Journal of Artificial Intelligence Research}}
   \bibinfo{volume}{29} (\bibinfo{year}{2007}), \bibinfo{pages}{49--77}.
\newblock
\urldef\tempurl%
\url{https://doi.org/10.1613/jair.2169}
\showDOI{\tempurl}


\bibitem[\protect\citeauthoryear{Berry}{Berry}{2000a}]%
        {berry-esterel-2000}
\bibfield{author}{\bibinfo{person}{G\'erard Berry}.}
  \bibinfo{year}{2000}\natexlab{a}.
\newblock \bibinfo{booktitle}{\emph{The Esterel v5 language primer: version
  v5\_91}}.
\newblock \bibinfo{publisher}{Centre de math\'ematiques appliqu\'ees, Ecole des
  mines and {INRIA}}.
\newblock
\urldef\tempurl%
\url{ftp://ftp-sop.inria.fr/marelle/Laurent.Thery/esterel/esterel.pdf}
\showURL{%
\tempurl}


\bibitem[\protect\citeauthoryear{Berry}{Berry}{2000b}]%
        {esterel}
\bibfield{author}{\bibinfo{person}{G{\'e}rard Berry}.}
  \bibinfo{year}{2000}\natexlab{b}.
\newblock \showarticletitle{The Foundations of Esterel}.
\newblock In \bibinfo{booktitle}{\emph{Proof, Language, and Interaction: Essays
  in Honour of Robin Milner}}, \bibfield{editor}{\bibinfo{person}{Gordon
  Plotkin}, \bibinfo{person}{Colin Stirling}, {and} \bibinfo{person}{Mads
  Tofte}} (Eds.). \bibinfo{publisher}{MIT Press}, \bibinfo{address}{Cambridge,
  MA, USA}, \bibinfo{pages}{425--454}.
\newblock
\showISBNx{9780262161886}


\bibitem[\protect\citeauthoryear{Berry}{Berry}{2002}]%
        {berry-constructive-2002}
\bibfield{author}{\bibinfo{person}{Gerard Berry}.}
  \bibinfo{year}{2002}\natexlab{}.
\newblock \showarticletitle{The Constructive Semantics of Pure {E}sterel. Draft
  version 3}.
\newblock  (\bibinfo{year}{2002}).
\newblock


\bibitem[\protect\citeauthoryear{Birkhoff}{Birkhoff}{1967}]%
        {birkhoff-lattice-1967}
\bibfield{author}{\bibinfo{person}{Garrett Birkhoff}.}
  \bibinfo{year}{1967}\natexlab{}.
\newblock \bibinfo{booktitle}{\emph{Lattice {Theory}} (\bibinfo{edition}{3rd}
  ed.)}. \bibinfo{series}{{AMS} {Colloquium} {Publications}},
  Vol.~\bibinfo{volume}{XXV}.
\newblock \bibinfo{publisher}{American Mathematical Society}.
\newblock


\bibitem[\protect\citeauthoryear{Boussinot and Susini}{Boussinot and
  Susini}{2000}]%
        {boussinot-java-2000}
\bibfield{author}{\bibinfo{person}{Frédéric Boussinot} {and}
  \bibinfo{person}{Jean-Ferdy Susini}.} \bibinfo{year}{2000}\natexlab{}.
\newblock \showarticletitle{Java threads and {SugarCubes}}.
\newblock \bibinfo{journal}{\emph{Software: Practice and Experience}}
  \bibinfo{volume}{30}, \bibinfo{number}{5} (\bibinfo{year}{2000}),
  \bibinfo{pages}{545--566}.
\newblock
\showISSN{1097-024X}
\urldef\tempurl%
\url{https://doi.org/10.1002/(SICI)1097-024X(20000425)30:5<545::AID-SPE308>3.0.CO;2-Q}
\showDOI{\tempurl}


\bibitem[\protect\citeauthoryear{Codognet and Diaz}{Codognet and Diaz}{1996}]%
        {codognet-compiling-1996}
\bibfield{author}{\bibinfo{person}{Philippe Codognet} {and}
  \bibinfo{person}{Daniel Diaz}.} \bibinfo{year}{1996}\natexlab{}.
\newblock \showarticletitle{Compiling constraints in clp(FD)}.
\newblock \bibinfo{journal}{\emph{The Journal of Logic Programming}}
  \bibinfo{volume}{27}, \bibinfo{number}{3} (\bibinfo{year}{1996}),
  \bibinfo{pages}{185 -- 226}.
\newblock
\showISSN{0743-1066}
\urldef\tempurl%
\url{https://doi.org/10.1016/0743-1066(95)00121-2}
\showDOI{\tempurl}


\bibitem[\protect\citeauthoryear{Crampton and Loizou}{Crampton and
  Loizou}{2001}]%
        {crampton-completion-2001}
\bibfield{author}{\bibinfo{person}{Jason Crampton} {and}
  \bibinfo{person}{George Loizou}.} \bibinfo{year}{2001}\natexlab{}.
\newblock \showarticletitle{The completion of a poset in a lattice of
  antichains}.
\newblock \bibinfo{journal}{\emph{International Mathematical Journal}}
  \bibinfo{volume}{1}, \bibinfo{number}{3} (\bibinfo{year}{2001}),
  \bibinfo{pages}{223--238}.
\newblock
\urldef\tempurl%
\url{http://www.isg.rhul.ac.uk/~jason/Pubs/imj.pdf}
\showURL{%
\tempurl}


\bibitem[\protect\citeauthoryear{Davey and Priestley}{Davey and
  Priestley}{2002}]%
        {davey-introduction-2002}
\bibfield{author}{\bibinfo{person}{B.~A. Davey} {and} \bibinfo{person}{H.~A.
  Priestley}.} \bibinfo{year}{2002}\natexlab{}.
\newblock \bibinfo{booktitle}{\emph{Introduction to lattices and order}}.
\newblock \bibinfo{publisher}{Cambridge University Press}.
\newblock
\showISBNx{0-521-78451-4 978-0-521-78451-1}


\bibitem[\protect\citeauthoryear{Diaz, Abreu, and Codognet}{Diaz
  et~al\mbox{.}}{2012}]%
        {diaz-implementation-2012}
\bibfield{author}{\bibinfo{person}{Daniel Diaz}, \bibinfo{person}{Salvador
  Abreu}, {and} \bibinfo{person}{Philippe Codognet}.}
  \bibinfo{year}{2012}\natexlab{}.
\newblock \showarticletitle{On the implementation of {GNU} Prolog}.
\newblock \bibinfo{journal}{\emph{Theory and Practice of Logic Programming}}
  \bibinfo{volume}{12}, \bibinfo{number}{1} (\bibinfo{year}{2012}),
  \bibinfo{pages}{253--282}.
\newblock
\showISSN{1471-0684, 1475-3081}
\urldef\tempurl%
\url{https://doi.org/10.1017/S1471068411000470}
\showDOI{\tempurl}


\bibitem[\protect\citeauthoryear{Fern\'{a}ndez and Hill}{Fern\'{a}ndez and
  Hill}{2004}]%
        {Fernandez:2004:ICS:963778.963779}
\bibfield{author}{\bibinfo{person}{Antonio~J. Fern\'{a}ndez} {and}
  \bibinfo{person}{Patricia~M. Hill}.} \bibinfo{year}{2004}\natexlab{}.
\newblock \showarticletitle{An Interval Constraint System for Lattice Domains}.
\newblock \bibinfo{journal}{\emph{ACM Trans. Program. Lang. Syst.}}
  \bibinfo{volume}{26}, \bibinfo{number}{1} (\bibinfo{date}{Jan.}
  \bibinfo{year}{2004}), \bibinfo{pages}{1--46}.
\newblock
\showISSN{0164-0925}
\urldef\tempurl%
\url{https://doi.org/10.1145/963778.963779}
\showDOI{\tempurl}


\bibitem[\protect\citeauthoryear{Gem\"unde, Brandt, and Schneider}{Gem\"unde
  et~al\mbox{.}}{2013}]%
        {gemunde-clock-2013}
\bibfield{author}{\bibinfo{person}{Mike Gem\"unde}, \bibinfo{person}{Jens
  Brandt}, {and} \bibinfo{person}{Klaus Schneider}.}
  \bibinfo{year}{2013}\natexlab{}.
\newblock \showarticletitle{Clock refinement in imperative synchronous
  languages}.
\newblock \bibinfo{journal}{\emph{{EURASIP} Journal on Embedded Systems}}
  \bibinfo{volume}{2013}, \bibinfo{number}{1} (\bibinfo{date}{10 Apr}
  \bibinfo{year}{2013}), \bibinfo{pages}{3}.
\newblock
\showISSN{1687-3963}
\urldef\tempurl%
\url{https://doi.org/10.1186/1687-3963-2013-3}
\showDOI{\tempurl}


\bibitem[\protect\citeauthoryear{Halbwachs}{Halbwachs}{1993}]%
        {synchronous-Halbwachs}
\bibfield{author}{\bibinfo{person}{Nicolas Halbwachs}.}
  \bibinfo{year}{1993}\natexlab{}.
\newblock \bibinfo{booktitle}{\emph{Synchronous programming of reactive
  systems}}.
\newblock \bibinfo{publisher}{Springer Science \& Business Media}.
\newblock
\showISBNx{9781441951335}
\urldef\tempurl%
\url{https://doi.org/10.1007/978-1-4757-2231-4}
\showDOI{\tempurl}


\bibitem[\protect\citeauthoryear{Harvey and Ginsberg}{Harvey and
  Ginsberg}{1995}]%
        {harvey-limited-1995}
\bibfield{author}{\bibinfo{person}{William~D. Harvey} {and}
  \bibinfo{person}{Matthew~L. Ginsberg}.} \bibinfo{year}{1995}\natexlab{}.
\newblock \showarticletitle{Limited discrepancy search}. In
  \bibinfo{booktitle}{\emph{Proceedings of the 14th International Conference on
  Artificial Intelligence}} \emph{(\bibinfo{series}{{IJCAI} '95})}.
  \bibinfo{pages}{607--615}.
\newblock
\urldef\tempurl%
\url{https://www.ijcai.org/Proceedings/95-1/Papers/080.pdf}
\showURL{%
\tempurl}


\bibitem[\protect\citeauthoryear{Jaffar and Lassez}{Jaffar and Lassez}{1987}]%
        {jaffar-constraint-1987}
\bibfield{author}{\bibinfo{person}{J. Jaffar} {and} \bibinfo{person}{J.-L.
  Lassez}.} \bibinfo{year}{1987}\natexlab{}.
\newblock \showarticletitle{Constraint Logic Programming}. In
  \bibinfo{booktitle}{\emph{Proceedings of the 14th ACM SIGACT-SIGPLAN
  Symposium on Principles of Programming Languages}}
  \emph{(\bibinfo{series}{POPL '87})}. \bibinfo{publisher}{ACM},
  \bibinfo{address}{New York, NY, USA}, \bibinfo{pages}{111--119}.
\newblock
\showISBNx{0-89791-215-2}
\urldef\tempurl%
\url{https://doi.org/10.1145/41625.41635}
\showDOI{\tempurl}


\bibitem[\protect\citeauthoryear{Janson and Haridi}{Janson and Haridi}{1991}]%
        {jaha91}
\bibfield{author}{\bibinfo{person}{Sverker Janson} {and} \bibinfo{person}{Seif
  Haridi}.} \bibinfo{year}{1991}\natexlab{}.
\newblock \showarticletitle{Programming {Paradigms} of the {Andorra} {Kernel}
  {Language}}. In \bibinfo{booktitle}{\emph{Logic {Programming}: {Proceedings}
  of the 1991 {International} {Logic} {Programming} {Symposium}}}.
  \bibinfo{publisher}{MIT Press}, \bibinfo{address}{San Diego, California}.
\newblock
\urldef\tempurl%
\url{http://eprints.sics.se/2097/}
\showURL{%
\tempurl}
\newblock
\shownote{(Revised version of SICS Research Report R91:08).}


\bibitem[\protect\citeauthoryear{Korf}{Korf}{1985}]%
        {Korf85depth-firstiterative-deepening}
\bibfield{author}{\bibinfo{person}{Richard~E. Korf}.}
  \bibinfo{year}{1985}\natexlab{}.
\newblock \showarticletitle{Depth-first iterative-deepening: An optimal
  admissible tree search}.
\newblock \bibinfo{journal}{\emph{Artificial Intelligence}}
  \bibinfo{volume}{27}, \bibinfo{number}{1} (\bibinfo{year}{1985}),
  \bibinfo{pages}{97 -- 109}.
\newblock
\showISSN{0004-3702}
\urldef\tempurl%
\url{https://doi.org/10.1016/0004-3702(85)90084-0}
\showDOI{\tempurl}


\bibitem[\protect\citeauthoryear{Korf}{Korf}{1996}]%
        {Korf96improvedlimited}
\bibfield{author}{\bibinfo{person}{Richard~E. Korf}.}
  \bibinfo{year}{1996}\natexlab{}.
\newblock \showarticletitle{Improved Limited Discrepancy Search}. In
  \bibinfo{booktitle}{\emph{Proceedings of the Thirteenth National Conference
  on Artificial Intelligence - Volume 1}} \emph{(\bibinfo{series}{AAAI'96})}.
  \bibinfo{publisher}{AAAI Press}, \bibinfo{pages}{286--291}.
\newblock
\showISBNx{0-262-51091-X}
\urldef\tempurl%
\url{http://new.aaai.org/Papers/AAAI/1996/AAAI96-043.pdf}
\showURL{%
\tempurl}


\bibitem[\protect\citeauthoryear{Laburthe and Caseau}{Laburthe and
  Caseau}{2002}]%
        {salsa-2002}
\bibfield{author}{\bibinfo{person}{Fran{\c{c}}ois Laburthe} {and}
  \bibinfo{person}{Yves Caseau}.} \bibinfo{year}{2002}\natexlab{}.
\newblock \showarticletitle{{SALSA}: A Language for Search Algorithms}.
\newblock \bibinfo{journal}{\emph{Constraints}} \bibinfo{volume}{7},
  \bibinfo{number}{3} (\bibinfo{date}{01 Jul} \bibinfo{year}{2002}),
  \bibinfo{pages}{255--288}.
\newblock
\showISSN{1572-9354}
\urldef\tempurl%
\url{https://doi.org/10.1023/A:1020565317875}
\showDOI{\tempurl}


\bibitem[\protect\citeauthoryear{Lee}{Lee}{2006}]%
        {lee-threads}
\bibfield{author}{\bibinfo{person}{Edward~A. Lee}.}
  \bibinfo{year}{2006}\natexlab{}.
\newblock \showarticletitle{The {Problem} with {Threads}}.
\newblock \bibinfo{journal}{\emph{Computer}} \bibinfo{volume}{39},
  \bibinfo{number}{5} (\bibinfo{date}{May} \bibinfo{year}{2006}),
  \bibinfo{pages}{33--42}.
\newblock
\showISSN{0018-9162}
\urldef\tempurl%
\url{https://doi.org/10.1109/MC.2006.180}
\showDOI{\tempurl}


\bibitem[\protect\citeauthoryear{Mandel}{Mandel}{2006}]%
        {mandel-conception-2006}
\bibfield{author}{\bibinfo{person}{Louis Mandel}.}
  \bibinfo{year}{2006}\natexlab{}.
\newblock \bibinfo{title}{Conception, S\'{e}mantique et Implantation de
  {ReactiveML} : un langage \`{a} la {ML} pour la programmation r\'{e}active}.
\newblock
\newblock


\bibitem[\protect\citeauthoryear{Mandel, Pasteur, and Pouzet}{Mandel
  et~al\mbox{.}}{2015}]%
        {mandel-time-2015}
\bibfield{author}{\bibinfo{person}{Louis Mandel}, \bibinfo{person}{C\'edric
  Pasteur}, {and} \bibinfo{person}{Marc Pouzet}.}
  \bibinfo{year}{2015}\natexlab{}.
\newblock \showarticletitle{Time refinement in a functional synchronous
  language}.
\newblock \bibinfo{journal}{\emph{Science of Computer Programming}}
  \bibinfo{volume}{111} (\bibinfo{year}{2015}), \bibinfo{pages}{190--211}.
\newblock
\urldef\tempurl%
\url{https://doi.org/10.1016/j.scico.2015.07.002}
\showDOI{\tempurl}


\bibitem[\protect\citeauthoryear{Mandel and Pouzet}{Mandel and Pouzet}{2005}]%
        {mandel-reactiveml:-2005}
\bibfield{author}{\bibinfo{person}{Louis Mandel} {and} \bibinfo{person}{Marc
  Pouzet}.} \bibinfo{year}{2005}\natexlab{}.
\newblock \showarticletitle{ReactiveML: A Reactive Extension to ML}. In
  \bibinfo{booktitle}{\emph{Proceedings of the 7th ACM SIGPLAN International
  Conference on Principles and Practice of Declarative Programming}}
  \emph{(\bibinfo{series}{PPDP '05})}. \bibinfo{publisher}{ACM},
  \bibinfo{address}{New York, NY, USA}, \bibinfo{pages}{82--93}.
\newblock
\showISBNx{1-59593-090-6}
\urldef\tempurl%
\url{https://doi.org/10.1145/1069774.1069782}
\showDOI{\tempurl}


\bibitem[\protect\citeauthoryear{Martinez, Fages, and Soliman}{Martinez
  et~al\mbox{.}}{2015}]%
        {martinez-search-2015}
\bibfield{author}{\bibinfo{person}{Thierry Martinez},
  \bibinfo{person}{François Fages}, {and} \bibinfo{person}{Sylvain Soliman}.}
  \bibinfo{year}{2015}\natexlab{}.
\newblock \showarticletitle{Search by constraint propagation}. In
  \bibinfo{booktitle}{\emph{Proceedings of the 17th International Symposium on
  Principles and Practice of Declarative Programming}}
  \emph{(\bibinfo{series}{{PPDP} '15})}. \bibinfo{publisher}{{ACM} Press},
  \bibinfo{pages}{173--183}.
\newblock
\showISBNx{978-1-4503-3516-4}
\urldef\tempurl%
\url{https://doi.org/10.1145/2790449.2790527}
\showDOI{\tempurl}


\bibitem[\protect\citeauthoryear{Michel and Van~Hentenryck}{Michel and
  Van~Hentenryck}{1999}]%
        {michel-localizer:-1999}
\bibfield{author}{\bibinfo{person}{Laurent Michel} {and}
  \bibinfo{person}{Pascal Van~Hentenryck}.} \bibinfo{year}{1999}\natexlab{}.
\newblock \showarticletitle{Localizer: {A} modeling language for local search}.
\newblock \bibinfo{journal}{\emph{INFORMS Journal on Computing}}
  \bibinfo{volume}{11}, \bibinfo{number}{1} (\bibinfo{year}{1999}),
  \bibinfo{pages}{1--14}.
\newblock
\urldef\tempurl%
\url{https://doi.org/10.1287/ijoc.11.1.1}
\showDOI{\tempurl}


\bibitem[\protect\citeauthoryear{Pelleau, Min\'e, Truchet, and
  Benhamou}{Pelleau et~al\mbox{.}}{2013}]%
        {pelleau-constraint-2013}
\bibfield{author}{\bibinfo{person}{Marie Pelleau}, \bibinfo{person}{Antoine
  Min\'e}, \bibinfo{person}{Charlotte Truchet}, {and}
  \bibinfo{person}{Fr\'ed\'eric Benhamou}.} \bibinfo{year}{2013}\natexlab{}.
\newblock \showarticletitle{A constraint solver based on abstract domains}. In
  \bibinfo{booktitle}{\emph{Verification, Model Checking, and Abstract
  Interpretation}}. \bibinfo{publisher}{Springer}, \bibinfo{pages}{434--454}.
\newblock
\urldef\tempurl%
\url{https://doi.org/10.1007/978-3-642-35873-9_26}
\showDOI{\tempurl}


\bibitem[\protect\citeauthoryear{Potop-Butucaru, Edwards, and
  Berry}{Potop-Butucaru et~al\mbox{.}}{2007}]%
        {esterel-compilation}
\bibfield{author}{\bibinfo{person}{Dumitru Potop-Butucaru},
  \bibinfo{person}{Stephen~A. Edwards}, {and} \bibinfo{person}{G\'erard
  Berry}.} \bibinfo{year}{2007}\natexlab{}.
\newblock \bibinfo{booktitle}{\emph{Compiling Esterel} (\bibinfo{edition}{1st}
  ed.)}.
\newblock \bibinfo{publisher}{Springer Publishing Company, Incorporated}.
\newblock
\urldef\tempurl%
\url{https://doi.org/10.1007%2F978-0-387-70628-3}
\showDOI{\tempurl}


\bibitem[\protect\citeauthoryear{Prosser and Unsworth}{Prosser and
  Unsworth}{2011}]%
        {prosser-limited-2011}
\bibfield{author}{\bibinfo{person}{Patrick Prosser} {and}
  \bibinfo{person}{Chris Unsworth}.} \bibinfo{year}{2011}\natexlab{}.
\newblock \showarticletitle{Limited discrepancy search revisited}.
\newblock \bibinfo{journal}{\emph{Journal of Experimental Algorithmics}}
  \bibinfo{volume}{16} (\bibinfo{date}{May} \bibinfo{year}{2011}),
  \bibinfo{pages}{1.1}.
\newblock
\showISSN{10846654}
\urldef\tempurl%
\url{https://doi.org/10.1145/1963190.2019581}
\showDOI{\tempurl}


\bibitem[\protect\citeauthoryear{Prud'homme, Fages, and Lorca}{Prud'homme
  et~al\mbox{.}}{2017}]%
        {choco}
\bibfield{author}{\bibinfo{person}{Charles Prud'homme},
  \bibinfo{person}{Jean-Guillaume Fages}, {and} \bibinfo{person}{Xavier
  Lorca}.} \bibinfo{year}{2017}\natexlab{}.
\newblock \bibinfo{booktitle}{\emph{Choco Documentation}}.
\newblock TASC - LS2N CNRS UMR 6241, COSLING S.A.S.
\newblock
\urldef\tempurl%
\url{http://www.choco-solver.org}
\showURL{%
\tempurl}


\bibitem[\protect\citeauthoryear{Rendl, Guns, Stuckey, and Tack}{Rendl
  et~al\mbox{.}}{2015}]%
        {rendl-minisearch:2015}
\bibfield{author}{\bibinfo{person}{Andrea Rendl}, \bibinfo{person}{Tias Guns},
  \bibinfo{person}{Peter~J. Stuckey}, {and} \bibinfo{person}{Guido Tack}.}
  \bibinfo{year}{2015}\natexlab{}.
\newblock \showarticletitle{{MiniSearch}: a solver-independent meta-search
  language for {MiniZinc}}. In \bibinfo{booktitle}{\emph{Principles and
  {Practice} of {Constraint} {Programming}}}. \bibinfo{publisher}{Springer},
  \bibinfo{pages}{376--392}.
\newblock
\urldef\tempurl%
\url{https://doi.org/10.1007/978-3-319-23219-5_27}
\showDOI{\tempurl}


\bibitem[\protect\citeauthoryear{Rossi, van Beek, and Walsh}{Rossi
  et~al\mbox{.}}{2006}]%
        {handbook-cp}
\bibfield{author}{\bibinfo{person}{Francesca Rossi}, \bibinfo{person}{Peter van
  Beek}, {and} \bibinfo{person}{Toby Walsh}.} \bibinfo{year}{2006}\natexlab{}.
\newblock \bibinfo{booktitle}{\emph{Handbook of Constraint Programming
  (Foundations of Artificial Intelligence)}}.
\newblock \bibinfo{publisher}{Elsevier Science Inc.}
\newblock


\bibitem[\protect\citeauthoryear{Saraswat, Gupta, and Jagadeesan}{Saraswat
  et~al\mbox{.}}{2014}]%
        {saraswat-tcc-2014}
\bibfield{author}{\bibinfo{person}{Vijay Saraswat}, \bibinfo{person}{Vineet
  Gupta}, {and} \bibinfo{person}{Radha Jagadeesan}.}
  \bibinfo{year}{2014}\natexlab{}.
\newblock \showarticletitle{{TCC}, with History}.
\newblock In \bibinfo{booktitle}{\emph{Horizons of the Mind. A Tribute to
  Prakash Panangaden}}. \bibinfo{publisher}{Springer},
  \bibinfo{pages}{458--475}.
\newblock
\urldef\tempurl%
\url{https://doi.org/10.1007/978-3-319-06880-0_24}
\showDOI{\tempurl}


\bibitem[\protect\citeauthoryear{Saraswat, Jagadeesan, and Gupta}{Saraswat
  et~al\mbox{.}}{1994}]%
        {tcc-lics94}
\bibfield{author}{\bibinfo{person}{Vijay Saraswat}, \bibinfo{person}{Radha
  Jagadeesan}, {and} \bibinfo{person}{Vineet Gupta}.}
  \bibinfo{year}{1994}\natexlab{}.
\newblock \showarticletitle{Foundations of Timed Concurrent Constraint
  Programming}. In \bibinfo{booktitle}{\emph{Proceedings Ninth Annual {IEEE}
  Symposium on Logic in Computer Science}}. \bibinfo{publisher}{IEEE},
  \bibinfo{pages}{71--80}.
\newblock
\urldef\tempurl%
\url{https://doi.org/10.1109/LICS.1994.316085}
\showDOI{\tempurl}


\bibitem[\protect\citeauthoryear{Saraswat, Jagadeesan, and Gupta}{Saraswat
  et~al\mbox{.}}{1996}]%
        {saraswat-timed-1999}
\bibfield{author}{\bibinfo{person}{Vijay Saraswat}, \bibinfo{person}{Radha
  Jagadeesan}, {and} \bibinfo{person}{Vineet Gupta}.}
  \bibinfo{year}{1996}\natexlab{}.
\newblock \showarticletitle{Timed Default Concurrent Constraint Programming}.
\newblock \bibinfo{journal}{\emph{Journal of Symbolic Computation}}
  \bibinfo{volume}{22}, \bibinfo{number}{5} (\bibinfo{year}{1996}),
  \bibinfo{pages}{475 -- 520}.
\newblock
\showISSN{0747-7171}
\urldef\tempurl%
\url{https://doi.org/10.1006/jsco.1996.0064}
\showDOI{\tempurl}


\bibitem[\protect\citeauthoryear{Saraswat, Jagadeesan, and Gupta}{Saraswat
  et~al\mbox{.}}{2003}]%
        {saraswat-jcc:-2003}
\bibfield{author}{\bibinfo{person}{Vijay Saraswat}, \bibinfo{person}{Radha
  Jagadeesan}, {and} \bibinfo{person}{Vineet Gupta}.}
  \bibinfo{year}{2003}\natexlab{}.
\newblock \showarticletitle{jcc: {Integrating} Timed Default Concurrent
  Constraint Programming into Java}.
\newblock In \bibinfo{booktitle}{\emph{Progress in {Artificial}
  {Intelligence}}}. \bibinfo{publisher}{Springer}, \bibinfo{pages}{156--170}.
\newblock
\urldef\tempurl%
\url{https://doi.org/10.1007/978-3-540-24580-3_23}
\showDOI{\tempurl}


\bibitem[\protect\citeauthoryear{Saraswat, Rinard, and Panangaden}{Saraswat
  et~al\mbox{.}}{1991}]%
        {saraswat-semantic-1991}
\bibfield{author}{\bibinfo{person}{Vijay~A. Saraswat}, \bibinfo{person}{Martin
  Rinard}, {and} \bibinfo{person}{Prakash Panangaden}.}
  \bibinfo{year}{1991}\natexlab{}.
\newblock \showarticletitle{The Semantic Foundations of Concurrent Constraint
  Programming}. In \bibinfo{booktitle}{\emph{Proceedings of the 18th ACM
  SIGPLAN-SIGACT Symposium on Principles of Programming Languages}}
  \emph{(\bibinfo{series}{POPL '91})}. \bibinfo{publisher}{ACM},
  \bibinfo{address}{New York, NY, USA}, \bibinfo{pages}{333--352}.
\newblock
\showISBNx{0-89791-419-8}
\urldef\tempurl%
\url{https://doi.org/10.1145/99583.99627}
\showDOI{\tempurl}


\bibitem[\protect\citeauthoryear{Schimpf and Shen}{Schimpf and Shen}{2012}]%
        {schimpf-eclipse-2012}
\bibfield{author}{\bibinfo{person}{Joachim Schimpf} {and} \bibinfo{person}{Kish
  Shen}.} \bibinfo{year}{2012}\natexlab{}.
\newblock \showarticletitle{{ECLiPSe} – From {LP} to {CLP}}.
\newblock \bibinfo{journal}{\emph{Theory and Practice of Logic Programming}}
  \bibinfo{volume}{12}, \bibinfo{number}{1} (\bibinfo{year}{2012}),
  \bibinfo{pages}{127--156}.
\newblock
\showISSN{1471-0684, 1475-3081}
\urldef\tempurl%
\url{https://doi.org/10.1017/S1471068411000469}
\showDOI{\tempurl}


\bibitem[\protect\citeauthoryear{Schrijvers, Demoen, Triska, and
  Desouter}{Schrijvers et~al\mbox{.}}{2014}]%
        {Schrijvers:2014:TMS:2608851.2608962}
\bibfield{author}{\bibinfo{person}{Tom Schrijvers}, \bibinfo{person}{Bart
  Demoen}, \bibinfo{person}{Markus Triska}, {and} \bibinfo{person}{Benoit
  Desouter}.} \bibinfo{year}{2014}\natexlab{}.
\newblock \showarticletitle{Tor: Modular Search with Hookable Disjunction}.
\newblock \bibinfo{journal}{\emph{Sci. Comput. Program.}}  \bibinfo{volume}{84}
  (\bibinfo{date}{May} \bibinfo{year}{2014}), \bibinfo{pages}{101--120}.
\newblock
\showISSN{0167-6423}
\urldef\tempurl%
\url{https://doi.org/10.1016/j.scico.2013.05.008}
\showDOI{\tempurl}


\bibitem[\protect\citeauthoryear{Schrijvers, Stuckey, and Wadler}{Schrijvers
  et~al\mbox{.}}{2009}]%
        {mcp}
\bibfield{author}{\bibinfo{person}{Tom Schrijvers}, \bibinfo{person}{Peter
  Stuckey}, {and} \bibinfo{person}{Philip Wadler}.}
  \bibinfo{year}{2009}\natexlab{}.
\newblock \showarticletitle{Monadic constraint programming}.
\newblock \bibinfo{journal}{\emph{Journal of Functional Programming}}
  \bibinfo{volume}{19}, \bibinfo{number}{06} (\bibinfo{date}{Nov.}
  \bibinfo{year}{2009}), \bibinfo{pages}{663--697}.
\newblock
\showISSN{0956-7968, 1469-7653}
\urldef\tempurl%
\url{https://doi.org/10.1017/S0956796809990086}
\showDOI{\tempurl}


\bibitem[\protect\citeauthoryear{Schrijvers, Tack, Wuille, Samulowitz, and
  Stuckey}{Schrijvers et~al\mbox{.}}{2013}]%
        {search2013}
\bibfield{author}{\bibinfo{person}{Tom Schrijvers}, \bibinfo{person}{Guido
  Tack}, \bibinfo{person}{Pieter Wuille}, \bibinfo{person}{Horst Samulowitz},
  {and} \bibinfo{person}{Peter~J. Stuckey}.} \bibinfo{year}{2013}\natexlab{}.
\newblock \showarticletitle{Search combinators}.
\newblock \bibinfo{journal}{\emph{Constraints}} \bibinfo{volume}{18},
  \bibinfo{number}{2} (\bibinfo{year}{2013}), \bibinfo{pages}{269--305}.
\newblock
\urldef\tempurl%
\url{https://doi.org/10.1007/s10601-012-9137-8}
\showDOI{\tempurl}


\bibitem[\protect\citeauthoryear{Schulte}{Schulte}{2002}]%
        {oz-space-schulte}
\bibfield{author}{\bibinfo{person}{Christian Schulte}.}
  \bibinfo{year}{2002}\natexlab{}.
\newblock \bibinfo{booktitle}{\emph{Programming Constraint Services: High-Level
  Programming of Standard and New Constraint Services}}.
  \bibinfo{series}{Lecture Notes in Computer Science},
  Vol.~\bibinfo{volume}{2302}.
\newblock \bibinfo{publisher}{Springer}.
\newblock
\urldef\tempurl%
\url{https://doi.org/10.1007/3-540-45945-6}
\showDOI{\tempurl}


\bibitem[\protect\citeauthoryear{Schulte and Stuckey}{Schulte and
  Stuckey}{2008}]%
        {schulte-efficient-2008}
\bibfield{author}{\bibinfo{person}{Christian Schulte} {and}
  \bibinfo{person}{Peter~J. Stuckey}.} \bibinfo{year}{2008}\natexlab{}.
\newblock \showarticletitle{Efficient {Constraint} {Propagation} {Engines}}.
\newblock \bibinfo{journal}{\emph{ACM Trans. Program. Lang. Syst.}}
  \bibinfo{volume}{31}, \bibinfo{number}{1} (\bibinfo{date}{Dec.}
  \bibinfo{year}{2008}), \bibinfo{pages}{2:1--2:43}.
\newblock
\showISSN{0164-0925}
\urldef\tempurl%
\url{https://doi.org/10.1145/1452044.1452046}
\showDOI{\tempurl}


\bibitem[\protect\citeauthoryear{Scott}{Scott}{1982}]%
        {scott-domains-1982}
\bibfield{author}{\bibinfo{person}{Dana~S. Scott}.}
  \bibinfo{year}{1982}\natexlab{}.
\newblock \showarticletitle{Domains for denotational semantics}. In
  \bibinfo{booktitle}{\emph{Automata, Languages and Programming}},
  \bibfield{editor}{\bibinfo{person}{Mogens Nielsen} {and}
  \bibinfo{person}{Erik~Meineche Schmidt}} (Eds.). \bibinfo{publisher}{Springer
  Berlin Heidelberg}, \bibinfo{address}{Berlin, Heidelberg},
  \bibinfo{pages}{577--610}.
\newblock
\urldef\tempurl%
\url{https://doi.org/10.1007/BFb0012801}
\showDOI{\tempurl}


\bibitem[\protect\citeauthoryear{Scott}{Scott}{2016}]%
        {scott-other-2016}
\bibfield{author}{\bibinfo{person}{Joseph Scott}.}
  \bibinfo{year}{2016}\natexlab{}.
\newblock \emph{\bibinfo{title}{Other {Things} {Besides} {Number}:
  {Abstraction}, {Constraint} {Propagation}, and {String} {Variable} {Types}}}.
\newblock \bibinfo{thesistype}{Ph.D. Dissertation}. \bibinfo{school}{Acta
  Universitatis Upsaliensis}, \bibinfo{address}{Uppsala}.
\newblock
\newblock
\shownote{OCLC: 943721122.}


\bibitem[\protect\citeauthoryear{Simonis and O'Sullivan}{Simonis and
  O'Sullivan}{2008}]%
        {Simonis:2008:SSR:1431540.1431546}
\bibfield{author}{\bibinfo{person}{Helmut Simonis} {and} \bibinfo{person}{Barry
  O'Sullivan}.} \bibinfo{year}{2008}\natexlab{}.
\newblock \showarticletitle{Search Strategies for Rectangle Packing}. In
  \bibinfo{booktitle}{\emph{Proceedings of the 14th International Conference on
  Principles and Practice of Constraint Programming}}
  \emph{(\bibinfo{series}{{CP} '08})}. \bibinfo{publisher}{Springer-Verlag},
  \bibinfo{address}{Berlin, Heidelberg}, \bibinfo{pages}{52--66}.
\newblock
\showISBNx{978-3-540-85957-4}
\urldef\tempurl%
\url{https://doi.org/10.1007/978-3-540-85958-1_4}
\showDOI{\tempurl}


\bibitem[\protect\citeauthoryear{Smyth}{Smyth}{1978}]%
        {smyth-power-1978}
\bibfield{author}{\bibinfo{person}{M.~B. Smyth}.}
  \bibinfo{year}{1978}\natexlab{}.
\newblock \showarticletitle{Power domains}.
\newblock \bibinfo{journal}{\emph{J. Comput. System Sci.}}
  \bibinfo{volume}{16}, \bibinfo{number}{1} (\bibinfo{year}{1978}),
  \bibinfo{pages}{23 -- 36}.
\newblock
\showISSN{0022-0000}
\urldef\tempurl%
\url{https://doi.org/10.1016/0022-0000(78)90048-X}
\showDOI{\tempurl}


\bibitem[\protect\citeauthoryear{Smyth, Motika, Rathlev, Hanxleden, and
  Mendler}{Smyth et~al\mbox{.}}{2017}]%
        {smyth-scest:-2018}
\bibfield{author}{\bibinfo{person}{Steven Smyth}, \bibinfo{person}{Christian
  Motika}, \bibinfo{person}{Karsten Rathlev}, \bibinfo{person}{Reinhard~Von
  Hanxleden}, {and} \bibinfo{person}{Michael Mendler}.}
  \bibinfo{year}{2017}\natexlab{}.
\newblock \showarticletitle{SCEst: Sequentially Constructive Esterel}.
\newblock \bibinfo{journal}{\emph{ACM Trans. Embed. Comput. Syst.}}
  \bibinfo{volume}{17}, \bibinfo{number}{2}, Article \bibinfo{articleno}{33}
  (\bibinfo{date}{Dec.} \bibinfo{year}{2017}), \bibinfo{numpages}{26}~pages.
\newblock
\showISSN{1539-9087}
\urldef\tempurl%
\url{https://doi.org/10.1145/3063129}
\showDOI{\tempurl}


\bibitem[\protect\citeauthoryear{Susini}{Susini}{2006}]%
        {susini-reactive-2006}
\bibfield{author}{\bibinfo{person}{Jean-Ferdy Susini}.}
  \bibinfo{year}{2006}\natexlab{}.
\newblock \showarticletitle{The Reactive Programming Approach on Top of
  Java/J2ME}. In \bibinfo{booktitle}{\emph{Proceedings of the 4th International
  Workshop on Java Technologies for Real-time and Embedded Systems}}
  \emph{(\bibinfo{series}{JTRES '06})}. \bibinfo{publisher}{ACM},
  \bibinfo{address}{New York, NY, USA}, \bibinfo{pages}{227--236}.
\newblock
\showISBNx{1-59593-544-4}
\urldef\tempurl%
\url{https://doi.org/10.1145/1167999.1168037}
\showDOI{\tempurl}


\bibitem[\protect\citeauthoryear{Tack}{Tack}{2009}]%
        {propagation-guido-tack}
\bibfield{author}{\bibinfo{person}{Guido Tack}.}
  \bibinfo{year}{2009}\natexlab{}.
\newblock \emph{\bibinfo{title}{Constraint {P}ropagation -- {M}odels,
  {T}echniques, {I}mplementation}}.
\newblock \bibinfo{thesistype}{Ph.D. Dissertation}. \bibinfo{school}{Saarland
  University}.
\newblock


\bibitem[\protect\citeauthoryear{Tardieu and Simone}{Tardieu and
  Simone}{2005}]%
        {tardieu-loops-2005}
\bibfield{author}{\bibinfo{person}{Olivier Tardieu} {and}
  \bibinfo{person}{Robert~de Simone}.} \bibinfo{year}{2005}\natexlab{}.
\newblock \showarticletitle{Loops in Esterel}.
\newblock \bibinfo{journal}{\emph{ACM Trans. Embed. Comput. Syst.}}
  \bibinfo{volume}{4}, \bibinfo{number}{4} (\bibinfo{date}{Nov.}
  \bibinfo{year}{2005}), \bibinfo{pages}{708--750}.
\newblock
\showISSN{1539-9087}
\urldef\tempurl%
\url{https://doi.org/10.1145/1113830.1113832}
\showDOI{\tempurl}


\bibitem[\protect\citeauthoryear{Technologies}{Technologies}{2005}]%
        {esterel-technologies-esterel-2005}
\bibfield{author}{\bibinfo{person}{Esterel Technologies}.}
  \bibinfo{year}{2005}\natexlab{}.
\newblock \bibinfo{booktitle}{\emph{The Esterel v7 Reference Manual Version v7
  30 -- initial {IEEE} standardization proposal}}.
\newblock \bibinfo{publisher}{Esterel Technologies}.
\newblock


\bibitem[\protect\citeauthoryear{Teppan, Friedrich, and Falkner}{Teppan
  et~al\mbox{.}}{2012}]%
        {teppan-quickpup:-2012}
\bibfield{author}{\bibinfo{person}{Erich~C. Teppan}, \bibinfo{person}{Gerhard
  Friedrich}, {and} \bibinfo{person}{Andreas~A. Falkner}.}
  \bibinfo{year}{2012}\natexlab{}.
\newblock \showarticletitle{QuickPup: A Heuristic Backtracking Algorithm for
  the Partner Units Configuration Problem}. In
  \bibinfo{booktitle}{\emph{Proceedings of the Twenty-Sixth AAAI Conference on
  Artificial Intelligence}} \emph{(\bibinfo{series}{AAAI'12})}.
  \bibinfo{publisher}{AAAI Press}, \bibinfo{pages}{2329--2334}.
\newblock
\urldef\tempurl%
\url{https://www.aaai.org/ocs/index.php/IAAI/IAAI-12/paper/viewPaper/4793}
\showURL{%
\tempurl}


\bibitem[\protect\citeauthoryear{Truchet and Assayag}{Truchet and
  Assayag}{2011}]%
        {truchet-constraint-2011}
\bibfield{author}{\bibinfo{person}{Charlotte Truchet} {and}
  \bibinfo{person}{G\'erard Assayag}.} \bibinfo{year}{2011}\natexlab{}.
\newblock \bibinfo{booktitle}{\emph{Constraint Programming in Music}}.
\newblock \bibinfo{publisher}{Wiley}.
\newblock
\showISBNx{9781848212886}
\showLCCN{2011008133}


\bibitem[\protect\citeauthoryear{Ueda}{Ueda}{2017}]%
        {ueda-logic/constraint-2017}
\bibfield{author}{\bibinfo{person}{Kazunori Ueda}.}
  \bibinfo{year}{2017}\natexlab{}.
\newblock \showarticletitle{Logic/Constraint Programming and Concurrency: The
  Hard-Won Lessons of the Fifth Generation Computer Project}.
\newblock \bibinfo{journal}{\emph{Science of Computer Programming}}
  (\bibinfo{year}{2017}).
\newblock
\showISSN{01676423}
\urldef\tempurl%
\url{https://doi.org/10.1016/j.scico.2017.06.002}
\showDOI{\tempurl}


\bibitem[\protect\citeauthoryear{Van~Hentenryck}{Van~Hentenryck}{1989}]%
        {van-hentenryck-constraint-chip-1989}
\bibfield{author}{\bibinfo{person}{Pascal Van~Hentenryck}.}
  \bibinfo{year}{1989}\natexlab{}.
\newblock \bibinfo{booktitle}{\emph{Constraint Satisfaction in Logic
  Programming}}.
\newblock \bibinfo{publisher}{{MIT} Press}.
\newblock
\showISBNx{0-262-08181-4}


\bibitem[\protect\citeauthoryear{Van~Hentenryck and Michel}{Van~Hentenryck and
  Michel}{2000}]%
        {van-hentenryck-opl-1999}
\bibfield{author}{\bibinfo{person}{Pascal Van~Hentenryck} {and}
  \bibinfo{person}{Laurant Michel}.} \bibinfo{year}{2000}\natexlab{}.
\newblock \showarticletitle{{OPL} Script: Composing and Controlling Models}. In
  \bibinfo{booktitle}{\emph{New Trends in Constraints}},
  \bibfield{editor}{\bibinfo{person}{Krzysztof~R. Apt}, \bibinfo{person}{Eric
  Monfroy}, \bibinfo{person}{Antonis~C. Kakas}, {and}
  \bibinfo{person}{Francesca Rossi}} (Eds.). \bibinfo{publisher}{Springer
  Berlin Heidelberg}, \bibinfo{address}{Berlin, Heidelberg},
  \bibinfo{pages}{75--90}.
\newblock
\showISBNx{978-3-540-44654-5}
\urldef\tempurl%
\url{https://doi.org/10.1007/3-540-44654-0_4}
\showDOI{\tempurl}


\bibitem[\protect\citeauthoryear{Van~Hentenryck and Michel}{Van~Hentenryck and
  Michel}{2005}]%
        {VanHentenryck:2005:CLS:1121598}
\bibfield{author}{\bibinfo{person}{Pascal Van~Hentenryck} {and}
  \bibinfo{person}{Laurent Michel}.} \bibinfo{year}{2005}\natexlab{}.
\newblock \bibinfo{booktitle}{\emph{Constraint-Based Local Search}}.
\newblock \bibinfo{publisher}{The MIT Press}.
\newblock
\showISBNx{0262220776}


\bibitem[\protect\citeauthoryear{Van~Hentenryck, Michel, and
  Liu}{Van~Hentenryck et~al\mbox{.}}{2005}]%
        {van-hentenryck-contraint-based-2005}
\bibfield{author}{\bibinfo{person}{Pascal Van~Hentenryck},
  \bibinfo{person}{Laurent Michel}, {and} \bibinfo{person}{Liyuan Liu}.}
  \bibinfo{year}{2005}\natexlab{}.
\newblock \showarticletitle{Contraint-based combinators for local search}.
\newblock \bibinfo{journal}{\emph{Constraints}} \bibinfo{volume}{10},
  \bibinfo{number}{4} (\bibinfo{year}{2005}), \bibinfo{pages}{363--384}.
\newblock
\showISSN{1572-9354}
\urldef\tempurl%
\url{https://doi.org/10.1007/s10601-005-2811-3}
\showDOI{\tempurl}


\bibitem[\protect\citeauthoryear{Van~Hentenryck, Perron, and
  Puget}{Van~Hentenryck et~al\mbox{.}}{2000}]%
        {van-hentenryck-search-2000}
\bibfield{author}{\bibinfo{person}{Pascal Van~Hentenryck},
  \bibinfo{person}{Laurent Perron}, {and} \bibinfo{person}{Jean-Fran\c{c}ois
  Puget}.} \bibinfo{year}{2000}\natexlab{}.
\newblock \showarticletitle{Search and Strategies in {OPL}}.
\newblock \bibinfo{journal}{\emph{ACM Trans. Comput. Logic}}
  \bibinfo{volume}{1}, \bibinfo{number}{2} (\bibinfo{date}{Oct.}
  \bibinfo{year}{2000}), \bibinfo{pages}{285--320}.
\newblock
\showISSN{1529-3785}
\urldef\tempurl%
\url{https://doi.org/10.1145/359496.359529}
\showDOI{\tempurl}


\bibitem[\protect\citeauthoryear{Walsh}{Walsh}{1997}]%
        {walsh-depth-bounded-1997}
\bibfield{author}{\bibinfo{person}{Toby Walsh}.}
  \bibinfo{year}{1997}\natexlab{}.
\newblock \showarticletitle{Depth-bounded discrepancy search}. In
  \bibinfo{booktitle}{\emph{Proceedings of the 15th International Conference on
  Artificial Intelligence}} \emph{(\bibinfo{series}{{IJCAI} '97})},
  Vol.~\bibinfo{volume}{97}. \bibinfo{pages}{1388--1393}.
\newblock
\urldef\tempurl%
\url{http://www.cse.unsw.edu.au/~tw/dds.pdf}
\showURL{%
\tempurl}

\end{thebibliography}
\end{document}